\documentclass[10pt,conference]{IEEEtran}

\usepackage{cite}
\usepackage{amsmath, amsthm, amsfonts}
\usepackage{amssymb}
\usepackage{hyperref}
\usepackage{color}
\usepackage{algorithm}
\usepackage{algpseudocode}
\usepackage{stmaryrd}
\usepackage{centernot}
\usepackage{graphicx}

\begin{document}

\title{Geometric Controller Synthesis for Weighted Event Signal
  Temporal Logic}

\author{\IEEEauthorblockN{Avinash Malik}
  \IEEEauthorblockA{\textit{Department of Electrical, Computer, and
      Software Engineering} \\
    \textit{University of Auckland}\\
    Auckland, New Zealand \\
    avinash.malik@auckland.ac.nz} }


\newtheorem{lemma}{Lemma}
\newtheorem{theorem}{Theorem}
\newtheorem{proposition}{Proposition}
\newtheorem{corollary}{Corollary}

\maketitle

\begin{abstract}
  Cyber-Physical Systems (CPS) controllers synthesized from standard
  temporal logics rely on rigid global clocks, rendering them vulnerable
  to asynchronous timing anomalies like clock snaps, jitter, and network
  delays. To overcome these vulnerabilities, we introduce a
  fundamentally timeless geometric control paradigm alongside a novel
  specification language: Weighted Event-Based Signal Temporal Logic
  (weSTL+). This logic combines the asynchronous, event-triggered
  mechanics of event-STL with the quantitative preference scaling of
  weighted-STL, making it highly suitable for practical autonomous CPS.
  Crucially, our framework compiles global weSTL+ specifications into a
  discrete, time-independent Hybrid Automaton (HA). Within this HA,
  operational modes transition strictly upon the crossing of continuous
  geometric state boundaries rather than rigid temporal triggers. By
  leveraging finite-time level-set inversion to map temporal windows
  into purely spatial surrogate constraints routed to a state-dependent
  quadratic program, this approach entirely eliminates explicit runtime
  clock monitoring. Autonomous robotics based experimental results
  demonstrate that the proposed geometric architecture guarantees the
  enforcement of safety and liveness under severe macroscopic timing
  discontinuities, succeeding where traditional time-indexed controllers
  fail.
\end{abstract}



\section{Introduction}
\label{sec:introduction}

Cyber-Physical Systems (CPS) tightly integrate computational algorithms
with physical processes, operating in highly dynamic environments where
control logic must enforce strict safety constraints while
simultaneously accommodating complex task
objectives~\cite{lee2008cyber}. Signal Temporal Logic
(STL)~\cite{maler2004monitoring} and its robust quantitative
semantics~\cite{donze2010robust} are traditionally used to formally
specify these requirements over continuous-time trajectories. However,
standard STL is inherently bound to rigid global clocks and strict
Boolean satisfaction criteria, making it insufficient for scenarios that
require asynchronous reactions or flexible, preferential trade-offs. To
encapsulate these complex requirements, specification languages have
evolved to address different functional gaps in CPS monitoring and
synthesis. For instance, Event-Based Signal Temporal Logic
(eSTL)~\cite{gundana2021event} provides asynchronous, event-anchored
evaluation mechanics that trigger specifications based on geometric
boundaries rather than rigid clocks. Conversely, Weighted Signal
Temporal Logic (wSTL)~\cite{mehdipour2021weighted} and its extension
(wSTL+)~\cite{cardona2023preferences} fulfill a different gap by
introducing a quantitative partial-satisfaction hierarchy, enabling the
formulation of inclusive (soft) specifications scaled by preference
weights.

While these logics successfully formalize CPS behaviors, synthesizing
downstream control architectures to enforce them remains a critical
challenge~\cite{yin2024formal}. Specifically, architectures such as
Signal Temporal Logic Model Predictive Control
(STL-MPC)~\cite{raman2014model} and Time-Varying Control Barrier
Functions (TV-CBF)~\cite{lindemann2018control} rely on explicit global
or local clock variables. We demonstrate that deploying these time-based
controllers introduces severe vulnerabilities, as their performance
assumes perfect clock synchrony. In realistic scenarios, the discrete
digital timeline frequently decouples from the physical state evolution,
leading to controller failures. To circumvent these theoretical and
practical vulnerabilities, this paper introduces a fundamentally
timeless control synthesis paradigm.

Our main contributions are summarized as follows:
\begin{itemize}
\item \textbf{A New STL Specification (weSTL+):} We introduce Weighted
  Event-Based Signal Temporal Logic (weSTL+), a unified grammar that
  synthesizes the asynchronous event anchors of
  eSTL~\cite{gundana2021event} with the quantitative preference scaling
  of wSTL+~\cite{cardona2023preferences}.
\item \textbf{Formal Semantics:} We define the formal two-level
  semantics of weSTL+, bridging global event-switched Boolean rules to
  continuous local clock robustness metrics.
\item \textbf{Sound Compiler Synthesis:} We develop a multi-pass
  compiler that directly translates the weSTL+ logic formulae into a
  sound but conservative execution framework.
\item \textbf{Geometric Control:} We implement a geometric verification
  architecture that entirely eliminates explicit runtime clock
  monitoring by transforming temporal logic specifications into
  continuous-time geometric inclusions via finite-time level-set
  inversion.
\item \textbf{Case Study:} We use a case-study of autonomous robotics,
  and its ablation, under timing anomalies to show that our geometric
  control technique is the only one that is able to overcome them in
  order to guarantee safety and liveness.
\end{itemize}

\section{Vulnerability of Time-Varying Control and Overview of the
  Proposed Solution}
\label{sec:vulnerability}

Control architectures that incorporate time directly into their
synthesis, such as Time-Varying Control Barrier Functions
(TV-CBF)~\cite{lindemann2018control,gundana2021event} and Signal
Temporal Logic Model Predictive Control (STL-MPC)~\cite{raman2014model},
operate under the implicit assumption of perfect global clock synchrony.
In single agent and distributed CPS, however, clocks inherently suffer
from skew and drift due to hardware
imperfections~\cite{lisova2017monitoring,cardoso2011network,brunelli2012temperature}
in addition to network delays.

To counteract these clock anomalies, network mechanisms like the Network
Time Protocol (NTP)~\cite{mills2010ntp} and the Precision Time Protocol
(PTP)~\cite{ieee1588_2008} are utilized. While PTP enables
sub-microsecond synchronization precision and NTP regulates broad
internet networks, their error-correction protocols often rely on
discrete clock steps/snaps, rate adjustments, and network jitters.
Importantly, these discrete time jumps are not confined to microscopic
scales; substantial evidence demonstrates that clock steps/snaps can
manifest on the order of seconds or
more~\cite{guo2012realtime,nist2018improving,ieee1588_2008}.

When controllers use these volatile time signals directly, macroscopic
timing anomalies decouple the digital timeline from the continuous
physical state evolution of the system.

\subsection{Motivating Example}
\label{sec:motiv-example-demons}

Our motivating example consists of an autonomous agent navigating a
two-dimensional environment populated with three static circular
obstacles and three dynamic moving hazards. The mission dictates that
the agent must depart a starting location, reach a search waypoint,
proceed to an inspection waypoint, and safely return to a designated
base. We evaluate two traditional time-indexed control methodologies for
this task: TV-CBF \cite{lindemann2018control} and STL-MPC
\cite{raman2014model}. The standard STL property governing these
baseline formulations requires reaching the first waypoint within
$5.5$~seconds, the second between $5.5$ and $10.5$~seconds, and the base
between $10.5$ and $15.5$~seconds, while strictly avoiding all spatial
hazards.

During the simulation, the system's global clock is subjected to three
distinct asynchronous timing anomalies: a discrete forward clock jump of
$+0.5$~seconds, a backward clock reset of $-0.8$~seconds, and a
completely frozen clock phase. Such
macroscopic, multi-second temporal discontinuities are a practical
reality across modern timing
architectures~\cite{liu2020automatic,nxp2025ieee1588}.

\begin{figure}[htbp]
\centering
\includegraphics[width=\columnwidth]{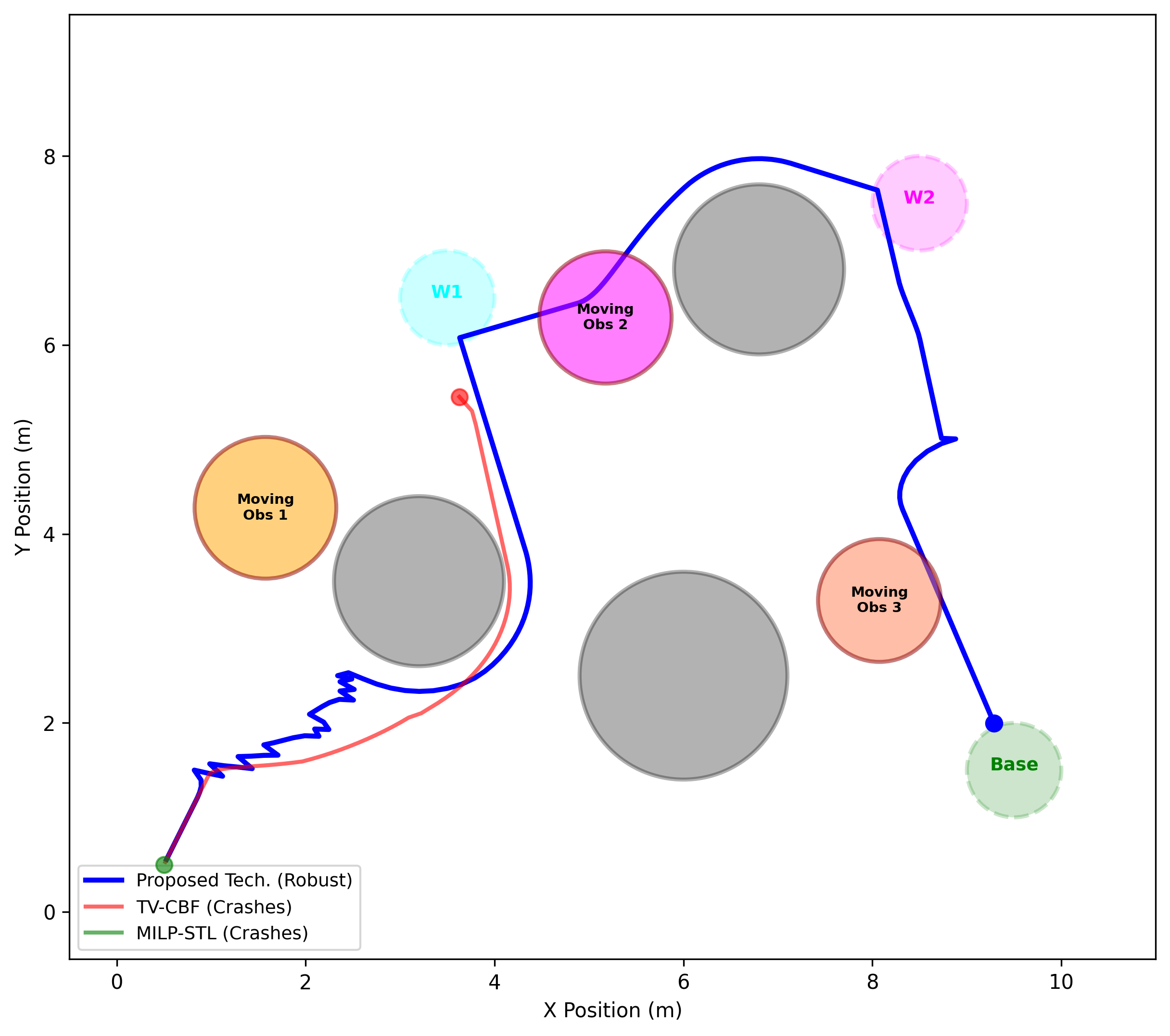}
\caption{Trajectory comparison of proposed technique, TV-CBF, and
  STL-MPC under asynchronous timing anomalies. Proposed geometric
  control trajectory going through the dynamic obstacles is not a
  collision these are so called ``ghost'' trajectories obstacles moving
  in place once the agent has passed.}
\label{fig:motivating_example}
\end{figure}

As depicted in Figure~\ref{fig:motivating_example}, both time-indexed
techniques fail catastrophically when subjected to these anomalies. The
TV-CBF architecture relies on a continuous, monotonic time evolution;
when the forward clock snap occurs, the local time bounds are instantly
violated, rendering the continuous-time constraints impossible to
satisfy and causing the controller to crash. Similarly, the STL-MPC
controller relies on a rigid discrete temporal grid. The sudden clock
jump permanently desynchronizes this discrete timeline from the physical
state evolution, which forces solver infeasibility, deadline violations,
and an ultimate system crash.

\subsection{Formalization of Clock Vulnerabilities in Current Control
  Architectures}
\label{sec:form-clock-vuln}

We now formalize how clock anomalies breaks the mathematical guarantees
of both continuous and discrete time-indexed controllers.

\begin{proposition}[Fragility of Time-Varying CBFs under Clock Perturbations]
\label{thm:tvcbf_fragility}
Let $h(x, t) \ge 0$ be a time-varying Control Barrier Function (TV-CBF).
We operate under the assumption that the TV-CBF is evaluated
continuously using the measured clock $\hat{t}(t) = t + \delta(t)$, and that
the barrier $h(x, \hat{t})$ is continuously differentiable with respect
to its temporal argument, consistent with standard continuous-time
safety formulations:
\begin{equation*}
    \nabla_x h(x, \hat{t})^T f(x, u) + \frac{\partial h(x, \hat{t})}{\partial \hat{t}}\dot{\hat{t}} \ge -\alpha(h(x, \hat{t}))
\end{equation*}
Because this formulation relies on the chain rule, if the timing anomaly
$\delta(t)$ exhibits a discontinuous jump (unbounded $\dot{\delta}(t)$), a
non-monotonic shift ($\dot{\delta}(t) < -1$), or induces a severe temporal
mismatch such that
$\frac{\partial h}{\partial \hat{t}}\big\vert_{\hat{t}} \gg \frac{\partial h}{\partial t}\big\vert_t$, the set
of admissible control inputs $\mathcal{U}_{\text{CBF}}(x, \hat{t})$ collapses or
forces premature stagnation.
\end{proposition}

\begin{proof}
  Under the stated assumption, the safety constraint depends strictly on
  the chain rule expansion of the corrupted time
  $\hat{t} = t + \delta(t)$, where the perceived clock rate is
  $\dot{\hat{t}} = 1 + \dot{\delta}(t)$.

\textbf{1. Clock Snaps (QP Infeasibility):} At a discontinuous snap
$t_s$, $\dot{\delta}(t)$ and consequently $\dot{\hat{t}}$ become unbounded.
Expanding the CBF constraint:
\begin{equation*}
  \dot{h}(x, \hat{t}) = L_f h(x, \hat{t}) + L_g h(x, \hat{t})u + \frac{\partial h}{\partial \hat{t}} \dot{\hat{t}} \ge -\alpha(h(x, \hat{t}))
\end{equation*}
Because $\dot{\hat{t}}$ is unbounded, satisfying this inequality
requires an unbounded control action. Since $u \in \mathcal{U}$ is strictly bounded,
the feasible control set instantly collapses
($\mathcal{U}_{\text{CBF}}(x, \hat{t}) = \emptyset$), rendering the QP solver infeasible.

\textbf{2. Phase Lag (Spurious Stalling):} For a continuous lag
$\delta(t) < 0$, the delay induces a spatio-temporal mismatch where the
system perceives a false safety violation:
$h(x, \hat{t}) \le 0 < h(x, t)$. Furthermore, if time shifts backwards
($\dot{\delta}(t) < -1$), the temporal gradient inverts. This triggers
maximum deceleration and stalls the vehicle, destroying liveness.
\end{proof}

\begin{proposition}[Loss of Liveness in Discrete STL-MPC under Timing
  Anomalies]
\label{thm:milp_liveness_loss}
Let $\phi = F_{[a, b]} \psi$ be a temporal liveness specification encoded over
a discrete time grid $t_k = k \Delta t$ ($k \in \{0, 1, \dots, N\}$) using
standard Mixed-Integer Linear Programming (MILP) Big-M bounds:
\begin{equation*}
  \mu(x(t_k)) + M (1 - z_k) \ge 0, \qquad \sum_{k=\lceil a/\Delta t \rceil}^{\lfloor b/\Delta t \rfloor} z_k \ge 1, \qquad z_k \in \{0, 1\}
\end{equation*}
Under asynchronous clock skew or execution delays $\delta(t)$, the discrete
timeline grid decouples from true physical state evolution, causing
irrevocable loss of liveness (solver infeasibility or goal reachability
stalling).
\end{proposition}

\begin{proof}
  Proof of Proposition~\ref{thm:milp_liveness_loss}. Let
  $T_{\text{reach}}$ denote the physical time at which continuous
  trajectory $x(t)$ satisfies predicate $\mu(x(T_{\text{reach}})) \ge 0$.
\begin{enumerate}
\item \textbf{Grid Index Shift (Missed Temporal Window):} Under ideal
  timing, $T_{\text{reach}} \in [a, b]$, yielding an index
  $k^* = \lfloor T_{\text{reach}}/\Delta t \rfloor \in [\lceil a/\Delta t \rceil, \lfloor b/\Delta t \rfloor]$ such that
  $z_{k^*} = 1$. Under time delay or clock drift $\delta_t$, the physical
  state reachability is delayed such that
  $T_{\text{reach}} + \delta_t > b$. The physical arrival occurs at discrete
  index $k_{\text{actual}} > \lfloor b/\Delta t \rfloor$.
    
\item \textbf{Infeasibility of Binary Sum Constraint:} Evaluating the
  discrete MILP constraint set over the original horizon requires
  $\sum_{k=\lceil a/\Delta t \rceil}^{\lfloor b/\Delta t \rfloor} z_k \ge 1$. However, for all grid indices
  $k \in [\lceil a/\Delta t \rceil, \lfloor b/\Delta t \rfloor]$, the physical state has not yet arrived
  ($\mu(x(t_k)) < 0$). Enforcing $z_k = 1$ forces a constraint violation:
  $ \mu(x(t_k)) + M(0) < 0 \implies \text{Violation of } \mu(x(t_k)) + M(1 -
  z_k) \ge 0 $. Consequently, all $z_k$ in the window are forced to $0$,
  yielding $\sum z_k = 0 < 1$.
\end{enumerate}
Because discrete MILP encodings bind temporal semantics strictly to
fixed sample indices $k$ rather than spatial capture basins, temporal
clock shifts break the logical feasibility set, permanently destroying
liveness guarantees.
\end{proof}

\subsection{Overview of the Proposed Solution}
\label{sec:proposed-solution}

To overcome the structural fragility of time-indexed control, we propose
Weighted Event-Based Signal Temporal Logic (weSTL+). This novel
specification language unifies asynchronous event anchors with
quantitative preference weights, allowing complex cyber-physical
missions to be defined by spatial boundaries and inclusive trade-offs
rather than rigid global clocks.

The proposed framework utilizes a compilation process to translate
weSTL+ formulae into a directly executable, time-independent control
policy. First, the global event-switched logic is compiled into a
discrete Hybrid Automaton (HA). Within this HA, operational modes
transition strictly upon the crossing of continuous geometric state
boundaries rather than temporal triggers. Second, for each active node
in the HA, local temporal specifications are recursively mapped into
smooth spatial surrogate fields. These purely geometric fields are then
routed into a unified, state-dependent Quadratic Program (QP).

Crucially, by leveraging finite-time level-set inversion to embed
bounded temporal horizons directly into contracted spatial boundaries,
explicit time is completely removed from both the HA mode transitions
and the compiled QP constraints. The resulting timeless geometric
controller acts entirely on the physical state space, maintaining safety
and liveness despite macroscopic clock snaps and freezes (as seen in
Figure~\ref{fig:motivating_example}). To establish the mathematical
foundation for this timeless compilation framework, the subsequent
section outlines the preliminary definitions governing system
trajectories and quantitative robustness.

\section{Preliminary Definitions}
\label{sec:base-defin-prel}

\paragraph{Global Continuous Trajectory:} Let
$x: \mathbb{R}_{\ge 0} \to \mathbb{R}^n$ be the continuous composite state trajectory
evaluated at global time $t \in \mathbb{R}_{\ge 0}$, where
$x(t) = [x_{\text{sys}}(t)^T, \theta(t)^T]^T$ encapsulates system dynamics
$x_{\text{sys}}(t)$ and dynamic environmental parameters $\theta(t)$.

\paragraph{Discrete Event Activation:} Let $\alpha$ be a discrete event label
with triggering time $t_\alpha \in \mathbb{R}_{\ge 0} \cup \{\infty\}$, defined via a state guard
predicate $\nu_{\alpha}: \mathbb{R}^n \to \mathbb{R}$:
$ t_\alpha \triangleq \inf \{ t \ge 0 \mid \nu_{\alpha}(x(t)) \ge 0 \} $, where
$t_\alpha = \infty$ if the guard set is never intersected (unactivated event).

\paragraph{Local Clock and Trajectory:} Upon event occurrence
$t_\alpha < \infty$, we define the local clock
$\tau \triangleq t - t_\alpha \ge 0$ and the time-shifted local trajectory
$x_\alpha: \mathbb{R}_{\ge 0} \to \mathbb{R}^n$ as: $ x_\alpha(\tau) \triangleq x(t_\alpha + \tau) $

\paragraph{Preference Weights:} Let $w \in \mathbb{R}_{>0}$ dictate the
importance/penalty scaling of inclusive (soft) specifications, and let
$\mathbf{w} = [w_1, \dots, w_m] \in \mathbb{R}_{>0}^m$ represent weight vectors for
inclusive connectives.

\paragraph{Quantitative Robustness Metrics:} Degree of specification
satisfaction is quantified via two real-valued satisfaction metrics:
\begin{itemize}
\item \textit{Local Robustness ($\rho_{\text{loc}}$):} A functional returning a scalar margin $\rho_{\text{loc}} \in \mathbb{R}$ for an unanchored local specification evaluated along a local trajectory $x_\alpha$ at local time $\tau$.
\item \textit{Global Robustness ($\rho$):} A functional returning an extended scalar margin $\rho \in \mathbb{R} \cup \{+\infty\}$ for an event-anchored global specification evaluated along a global trajectory $x$ at global time $t$.
\end{itemize}

\section{Formal Syntax of weSTL+}
\label{sec:formal-syntax-westl+}

Following eSTL~\cite{gundana2021event} and
wSTL+~\cite{cardona2023preferences}, the grammar is defined across two
distinct levels: Level 1 (Local STL Specifications) and Level 2 (Global
Event-Anchored weSTL+ Specifications).

\subsection{Level 1: Local STL Grammar}
Local specifications $\psi$ are standard unanchored STL formulae evaluated
on the local clock $\tau$:
\begin{align*} \psi ::= \text{True} \mid \mu \ge 0 \mid \neg \psi \mid \psi_1 \wedge \psi_2 \mid \psi_1 \vee \psi_2 \mid G_{[0,b]} \psi
  \mid F_{[0,b]} \psi
\end{align*}

where $\mu: \mathbb{R}^n \to \mathbb{R}$ is a continuously differentiable predicate function.
For environment-dependent specifications, $\mu(x)$ serves as shorthand for
$\mu(x_{\text{sys}}; \theta)$ evaluated on composite state
$x = [x_{\text{sys}}^T, \theta^T]^T$.

\subsection{Level 2: Global weSTL+ Grammar}

Let $\mathcal{E} = \{\alpha_1, \alpha_2, \dots\}$ be a set of discrete event anchors, where
each event $\alpha \in \mathcal{E}$ is defined by a continuous state guard predicate
$\nu_{\alpha}$. Global specifications $\phi$ incorporate event anchors
$\alpha \in \mathcal{E}$, logical modalities ($ex$ vs. $in$), and preference weights $w$:
\begin{equation*}
  \begin{aligned}
    \phi ::= \, &\text{True} \mid \neg \phi \mid \phi_1 \wedge_{ex} \phi_2 \mid \phi_1 \wedge_{in}^{\mathbf{w}} \phi_2 \mid \phi_1 \vee_{ex} \phi_2 \mid \phi_1 \vee_{in}^{\mathbf{w}} \phi_2 \\
             &\mid G_{ex, [0,b]}^{\alpha} \psi \mid F_{ex, [0,b]}^{\alpha} \psi \mid G_{in, [0,b]}^{\alpha, w} \psi \mid F_{in, [0,b]}^{\alpha, w} \psi
  \end{aligned}
\end{equation*}

\section{Two-Level Formal Semantics}
\label{sec:two-level-formal}

Following the two-level syntax, we define both Boolean and robust
semantics. We establish Boolean rules to dictate strict logical
satisfaction, while the quantitative robust semantics provide a
continuous margin of satisfaction essential for downstream
continuous-time control synthesis. When an event ($\alpha$) triggers, local
standard STL properties are evaluated on the local clock $\tau$, providing
a clean separation between high-level event semantics and low-level
continuous semantics.

\subsection{Level 1: Local Clock Semantics}
\label{sec:level-1:-local}

For a local trajectory $x_\alpha$ and local clock $\tau \ge 0$, the Boolean
semantics $\models_{\text{loc}}$ are defined as follows (where interval
$I = [\tau + 0, \tau + b]$):
{\small \begin{align*}
  (x_\alpha, \tau) &\models_{\text{loc}} \text{True} \iff \top \\
    (x_\alpha, \tau) &\models_{\text{loc}} (\mu \ge 0) \iff \mu(x_\alpha(\tau)) \ge 0 \\[-2pt]
    (x_\alpha, \tau) &\models_{\text{loc}} \neg \psi \iff \neg \big( (x_\alpha, \tau) \models_{\text{loc}} \psi \big) \\[-2pt]
    (x_\alpha, \tau) &\models_{\text{loc}} (\psi_1 \wedge \psi_2) \iff (x_\alpha, \tau) \models_{\text{loc}} \psi_1 \wedge (x_\alpha, \tau) \models_{\text{loc}} \psi_2 \\[-2pt]
    (x_\alpha, \tau) &\models_{\text{loc}} (\psi_1 \vee \psi_2) \iff (x_\alpha, \tau) \models_{\text{loc}} \psi_1 \vee (x_\alpha, \tau) \models_{\text{loc}} \psi_2 \\[-2pt]
    (x_\alpha, \tau) &\models_{\text{loc}} G_{[0,b]} \psi \iff \forall \tau' \in I, \, (x_\alpha, \tau') \models_{\text{loc}} \psi \\[-2pt]
    (x_\alpha, \tau) &\models_{\text{loc}} F_{[0,b]} \psi \iff \exists \tau' \in I \text{ s.t. } (x_\alpha, \tau') \models_{\text{loc}} \psi
\end{align*}
}
\noindent\textbf{Local Robustness ($\rho_{\text{loc}}$):}
\begin{align*}
  \rho_{\text{loc}}(\text{True}, x_\alpha, \tau) &\triangleq +\infty \\
    \rho_{\text{loc}}(\mu \ge 0, x_\alpha, \tau) &\triangleq \mu(x_\alpha(\tau)) \\[-2pt]
    \rho_{\text{loc}}(\neg \psi, x_\alpha, \tau) &\triangleq -\rho_{\text{loc}}(\psi, x_\alpha, \tau) \\[-2pt]
    \rho_{\text{loc}}(\psi_1 \wedge \psi_2, x_\alpha, \tau) &\triangleq \min \big( \rho_{\text{loc}}(\psi_1,x_{\alpha},\tau), \, \rho_{\text{loc}}(\psi_2,x_{\alpha},\tau) \big) \\[-2pt]
    \rho_{\text{loc}}(\psi_1 \vee \psi_2, x_\alpha, \tau) &\triangleq \max \big( \rho_{\text{loc}}(\psi_1,x_{\alpha},\tau), \, \rho_{\text{loc}}(\psi_2,x_{\alpha},\tau) \big) \\[-2pt]
    \rho_{\text{loc}}(G_{[0,b]} \psi, x_\alpha, \tau) &\triangleq \min_{\tau' \in I} \rho_{\text{loc}}(\psi, x_\alpha, \tau') \\[-2pt]
    \rho_{\text{loc}}(F_{[0,b]} \psi, x_\alpha, \tau) &\triangleq \max_{\tau' \in I} \rho_{\text{loc}}(\psi, x_\alpha, \tau')
\end{align*}

\subsection{Level 2: Global Event-Switched Semantics}

The global semantics map time $t$ to local specifications via event
times $t_\alpha$. Let $I' = [0,b]$. For $O \in \{G, F\}$, the specification is
trivially true if the event is dormant
($t < t_\alpha \vee t_\alpha = \infty$). Otherwise, it triggers local evaluation:
\begin{equation*}
    (x, t) \models O_{ex, I'}^{\alpha} \psi \iff (t < t_\alpha \vee t_\alpha = \infty) \vee \big((x_\alpha, 0) \models_{\text{loc}} O_{I'} \psi\big)
\end{equation*}

Inclusive operators ($O_{in}$) share identical Boolean truth values with
their exclusive counterparts. Standard Boolean semantics apply to
negation, conjunction, and inclusive disjunction:
\begin{align*}
  (x, t) &\models \text{True} \iff \top \\
    (x,t) &\models \neg \phi \iff \neg \big((x,t) \models \phi \big) \\[-2pt]
    (x,t) &\models \phi_1 \wedge_{ex} \phi_2 \iff \big((x,t) \models \phi_1\big) \wedge \big((x,t) \models \phi_2\big) \\[-2pt]
    (x,t) &\models \phi_1 \wedge_{in}^{\mathbf{w}} \phi_2 \iff \big((x,t) \models \phi_1\big) \wedge \big((x,t) \models \phi_2\big) \\[-2pt]
    (x,t) &\models \phi_1 \vee_{in}^{\mathbf{w}} \phi_2 \iff \big((x,t) \models \phi_1\big) \vee \big((x,t) \models \phi_2\big)
\end{align*}

\noindent Exclusive disjunction enforces mutual exclusivity:
\begin{equation*}
    (x,t) \models \phi_1 \vee_{ex} \phi_2 \iff \big((x,t) \models \phi_1 \wedge \neg \phi_2\big) \vee \big((x,t) \models \neg \phi_1 \wedge \phi_2\big)
\end{equation*}

\noindent\textbf{Global Quantitative Robustness ($\rho$):} \\
\noindent \textit{1. Boolean Connectives:}
\begin{align*}
  \rho(\text{True}, x, \tau) &\triangleq +\infty \\
    \rho(\neg \phi, x, t) &\triangleq -\rho(\phi, x, t) \\[-2pt]
  \rho(\phi_1 \wedge_{ex} \phi_2, x, t) &\triangleq \min \big( \rho(\phi_1, x, t), \, \rho(\phi_2, x, t) \big) \\[-2pt]
  \rho(\phi_1 \wedge_{in}^{\mathbf{w}} \phi_2, x, t) &\triangleq \min \big( w_1 \rho(\phi_1, x, t), \, w_2 \rho(\phi_2, x, t) \big) \\[-2pt]
    \rho(\phi_1 \vee_{ex} \phi_2, x, t) &\triangleq \min \big( \max(\rho(\phi_1, x, t), \rho(\phi_2, x, t)), \\[-2pt]
    &\qquad \qquad \;\; -\min(\rho(\phi_1, x, t), \rho(\phi_2, x, t)) \big) \\[-2pt]
    \rho(\phi_1 \vee_{in}^{\mathbf{w}} \phi_2, x, t) &\triangleq \max \big( w_1 \rho(\phi_1, x, t), \, w_2 \rho(\phi_2, x, t) \big)
\end{align*}

\noindent \textit{2. Event Operators (Exclusive and Inclusive):} \\
For $O \in \{G, F\}$, let weight $w'=1$ for exclusive modalities
($O_{ex}$) and $w'=w$ for inclusive modalities ($O_{in}$). Robustness
evaluates to $+\infty$ during the dormant phase
($t < t_\alpha \vee t_\alpha = \infty$). During the active phase
($t \ge t_\alpha \wedge t_\alpha < \infty$), it scales the local robustness:
\begin{equation*}
    \rho(O_{mod, I}^{\alpha, w'} \psi, \, x, \, t) \triangleq 
    \begin{cases}
        +\infty, & \text{dormant} \\
        w' \cdot \rho_{\text{loc}}(O_I \psi, x_\alpha, 0), & \text{active}
    \end{cases}
\end{equation*}

\section{Theoretical Guarantees: Soundness of Quantitative Semantics}
\label{sec:theor-guar-soundn}

To validate that quantitative robustness $\rho$ serves as a faithful metric
for formal specification checking, we establish the soundness of the
two-level quantitative semantics with respect to the Boolean semantics.

\subsection{Level 1 Soundness (Standard STL)}

Since Level 1 local specifications $\psi$ and local robustness metrics
$\rho_{\text{loc}}$ conform strictly to standard continuous-time
STL~\cite{donze2010robust}, Level 1 soundness follows directly from
established STL semantics.

\begin{lemma}[Level 1 Equivalence~\cite{donze2010robust,
    maler2004monitoring}]
  \label{lem:level1_soundness}
  For any local STL specification $\psi$, local trajectory $x_\alpha$, and local
  clock $\tau \ge 0$:
  $ \rho_{\text{loc}}(\psi, x_\alpha, \tau) \ge 0 \iff (x_\alpha, \tau) \models_{\text{loc}} \psi $.
\end{lemma}

\subsection{Level 2 Global Soundness}

We establish the soundness of the Level 2 quantitative robustness
semantics $\rho$ relative to the global Boolean evaluation $\models$ purely as a
property over arbitrary continuous trajectories $x$.

\begin{theorem}[Global Soundness for all weSTL+ Modalities]
\label{thm:global_soundness}
For any global weSTL+ specification $\phi$ generated by the Level 2
grammar, continuous trajectory $x$, and global time $t \ge 0$:
$ \rho(\phi, x, t) \ge 0 \iff (x,t) \models \phi $.
\end{theorem}

\begin{proof}
The proof proceeds by structural induction over the syntax of the Level 2 weSTL+ grammar.

\noindent\textbf{Base Cases:} 
For $\phi = \text{True}$,
$\rho(\text{True}) = +\infty \ge 0 \iff \text{True}$. 

\noindent\textbf{Inductive Step (Boolean Connectives):} 
Assume the equivalence holds for sub-formulas $\phi_1, \phi_2$. Since preference weights $w, \mathbf{w} > 0$ strictly preserve signs, the quantitative robustness aligns with the Boolean rules:
\begin{itemize}
\item \textit{Conjunctions ($\wedge_{ex}, \wedge_{in}^{\mathbf{w}}$):} Both
  $\min(\rho(\phi_{1}), \rho(\phi_{2})) \ge 0$ and
  $\min(w_1 \rho(\phi_{1}), w_2 \rho(\phi_{2})) \ge 0$ hold if and only if
  $\rho(\phi_{1}) \ge 0$ and $\rho(\phi_{2}) \ge 0$ (since
  $w_1, w_2 > 0$), satisfying the strict Boolean intersection.
\item \textit{Disjunctions ($\vee_{ex}, \vee_{in}^{\mathbf{w}}$):} For
  inclusive disjunction,\\
  \mbox{$\max(w_1 \rho_1, w_2 \rho_2) \ge 0 \iff \rho_1 \ge 0 \vee \rho_2 \ge 0$}, where
  $\rho_{1} = \rho(\phi_{1})$ and $\rho_{2} = \rho(\phi_{2})$. For exclusive disjunction, \\
  $\min(\max(\rho_1, \rho_2), -\min(\rho_1, \rho_2)) \ge 0$ guarantees that exactly
  one sub-formula is non-negative, matching strict mutual exclusivity.
    \item \textit{Negation ($\neg \phi$):} $\rho(\neg \phi) = -\rho(\phi) \ge 0 \iff \rho(\phi) \le 0 \iff (x,t) \models \neg \phi$.
\end{itemize}

\noindent\textbf{Inductive Step (Event Modalities):} 
Let $O \in \{G, F\}$ and let $w = 1$ for exclusive modalities ($O_{ex}$) or $w > 0$ for inclusive modalities ($O_{in}$). Evaluating across the event activation boundary $t_\alpha$ yields:
\begin{itemize}
    \item \textit{Dormant Phase ($t < t_\alpha$ or $t_\alpha = \infty$):} By definition, $\rho(\phi, x, t) = +\infty \ge 0$, which evaluates to $\text{True}$.
    \item \textit{Active Phase ($t \ge t_\alpha$ and
        $t_\alpha < \infty$):} By global robustness definition,
      $\rho(\phi, x, t) = w \cdot \rho_{\text{loc}}(O_{[0,b]} \psi, x_\alpha, 0)$. Applying
      weight sign-invariance alongside Level~1 STL soundness
      (Lemma~\ref{lem:level1_soundness}):
    \begin{align*}
      w \cdot \rho_{\text{loc}}(O_{[0,b]} \psi, x_\alpha, 0) \ge 0 
      &\iff \rho_{\text{loc}}(O_{[0,b]} \psi, x_\alpha, 0) \ge 0 \\
      &\iff (x_\alpha, 0) \models_{\text{loc}} O_{[0,b]} \psi \\
      &\iff (x,t) \models \phi
    \end{align*}
\end{itemize}
By structural induction, the soundness equivalence holds for all valid
weSTL+ specifications.
\end{proof}

\section{Compiling weSTL+ to Event Driven Geometric Control}
\label{sec:compiler-rules}

In this section we give the compilation steps for the proposed weSTL+
grammar. We first give the overall compilation approach in
Section~\ref{sec:overall-compilation-approach}. Followed by the crucial
technique of inverting time into spatial level sets in
Section~\ref{sec:cont-time-geom}. Finally, the detailed compilation
steps are described in
Sections~\ref{sec:local-stl-inductive}--~\ref{sec:ha-generation}.

\subsection{Overall Compilation Approach}
\label{sec:overall-compilation-approach}

The compilation of a global weSTL+ specification into an executable
event-driven control policy proceeds in two primary sequential
steps. First, the specification undergoes structural
validation, flattening, and spatial surrogate generation. Second, the
resulting time-independent surrogates are mapped into affine
differential constraints for optimization.

\textbf{Step 1: Structural Validation and Surrogate Generation} \\ The
compiler processes the logic through a strict functional composition
pipeline to produce the final surrogate constraint sets, defined as:
\begin{equation*}
  \hat{\rho} = \text{Compile}(\text{valid}_{\text{CNF}}(\text{CNF}(\text{Flatten}(\text{valid}(\phi)))))
\end{equation*}
\begin{itemize}
\item \textbf{Validation (valid):} The compiler verifies
  that the specification contains no alternating temporal nestings (e.g.,
  $F(G(\dots))$ or $G(F(\dots))$). If such nested heterogeneous operators
  are detected, the program is rejected to preserve the decoupled control
  synthesis architecture.

\item \textbf{Flattening (Flatten):} To mathematically preserve temporal
  semantics while strictly adhering to the two-level grammar, the
  flattening function recursively absorbs inner homogeneous Level 1
  operators and distributes temporal boundaries. To prevent over-flattening
  and preserve the efficiency of the continuous-time spatial mappings,
  temporal operators are only distributed across connectives if the inner
  specifications contain nested temporal horizons. Distributed Level 1
  connectives are universally elevated to their exclusive Level 2 counterparts
  ($\wedge_{ex}, \vee_{ex}$). For any temporal operator
  $O \in \{G, F\}$, modality $mod \in \{ex, in\}$, event anchor
  $\alpha$, weight $w$, Level 1 specification $\psi$, and Level 2 connective
  $\circ \in \{\wedge_{ex}, \wedge_{in}^{\mathbf{w}}, \vee_{ex}, \vee_{in}^{\mathbf{w}}\}$,
  recursive flattening is defined as:
  {\small \begin{align*}
    \text{Flatten}(\phi_1 \circ \phi_2) &= \text{Flatten}(\phi_1) \circ \text{Flatten}(\phi_2) \\
    \text{Flatten}\big(O_{mod, [0, a]}^{\alpha, w}(O_{[0, b]} \psi)\big) &= \text{Flatten}\big(O_{mod, [0, a+b]}^{\alpha, w}(\psi)\big) \\
    \text{Flatten}\big(G_{mod, [0, a]}^{\alpha, w}(\psi_1 \wedge \psi_2)\big) &= \text{Flatten}\big(G_{mod, [0, a]}^{\alpha, w}(\psi_1)\big) \\
    &\quad \wedge_{ex} \text{Flatten}\big(G_{mod, [0, a]}^{\alpha, w}(\psi_2)\big) \\
    \text{Flatten}\big(F_{mod, [0, a]}^{\alpha, w}(\psi_1 \vee \psi_2)\big) &= \text{Flatten}\big(F_{mod, [0, a]}^{\alpha, w}(\psi_1)\big) \\
    &\quad \vee_{ex} \text{Flatten}\big(F_{mod, [0, a]}^{\alpha, w}(\psi_2)\big) \\
    \text{Flatten}\big(O_{mod, [0, a]}^{\alpha, w}(\psi)\big) &= O_{mod, [0, a]}^{\alpha, w}(\psi) \\
    \text{Flatten}(\psi) &= \psi 
  \end{align*}
  }where the distributive rules for $G$ and $F$ apply strictly when $\psi_1$ or
  $\psi_2$ contain nested temporal bounds. The terminal base cases apply
  exclusively when the inner specification $\psi$ is devoid of internal
  temporal bounds, allowing pure unnested spatial propositions to be
  evaluated efficiently via smooth local surrogates. If the recursive pass
  encounters mathematically unflattenable combinations containing nested
  temporal horizons (e.g., $G$ over $\vee$, or $F$ over $\wedge$), the compiler
  throws a structural validity error, as downstream continuous-time spatial
  mappings cannot process trapped internal temporal horizons.

\item \textbf{Normalization (CNF):} To ensure uniform routing to the QP
  solver, the flattened boolean skeleton is converted into Conjunctive
  Normal Form (CNF): $\bigwedge_i (\bigvee_j \phi_{i,j})$. This guarantees that all
  disjunctions reside strictly below conjunctions.
\item \textbf{Post-CNF Validation ($\text{valid}_{\text{CNF}}$):} To
  preserve soundness when mapping logic to differential constraints, the
  compiler verifies that no inner disjunctive clause mixes heterogeneous
  temporal operators. If any clause contains both global safety ($G$)
  and liveness ($F$) operators (e.g., $G \vee F$), the program is rejected,
  as safety and liveness requirements exhibit fundamentally incompatible
  control semantics under disjunction.
\item \textbf{Surrogate Construction (Compile):} The normalized
  specification is mapped into differentiable spatial fields. As defined
  in Sections~\ref{sec:local-stl-inductive}
  and~\ref{sec:global-event-level}, this step generates the local
  inductive surrogates and applies global temporal horizon embeddings.
\end{itemize}

\textbf{Step 2: Control Mapping and Synthesis} \\ Once the spatial
surrogates are constructed, they are mapped into continuous control
laws:
\begin{itemize}
\item \textbf{Constraint Mapping:} As detailed in
  Section~\ref{sec:glob-temp-contr}, the active spatial surrogates are
  translated into corresponding Control Barrier Functions (CBFs) for
  safety and Control Lyapunov Functions (CLFs) for liveness, each
  augmented with appropriate preference-weighted slack variables.
\item \textbf{Quadratic Program (QP) Assembly:}
  Section~\ref{sec:unif-contr-synth} defines how these independent
  affine differential constraints are unified into a state-dependent QP,
  ensuring the objective function penalizes slack variables according to
  their specified preference weights.
\item \textbf{Hybrid Automaton (HA) Generation and Binding:} Finally, as
  detailed in Section~\ref{sec:ha-generation}, the global specification
  is compiled into a time-independent Hybrid Automaton. The
  event-specific QPs are bound to discrete operational modes
  $q \in \mathcal{Q}$. Because mode transitions are triggered exactly when the
  physical state crosses a continuous guard boundary
  ($\nu_{\alpha}(x(t)) \ge 0$), the active subsets of compiled QP constraints are
  reconfigured dynamically, executing the mission without reliance on a
  global clock.
\end{itemize}

\subsection{Continuous-Time Geometric Semantic Mappings}
\label{sec:cont-time-geom}

To eliminate the need for explicit global clocks during closed-loop
execution, we systematically map temporal windows into purely geometric
spatial constraints. This transformation is fundamentally anchored by
the Bhat-Bernstein theorem~\cite{bhat2000finite}, which originally
established that continuous dynamical systems can guarantee strict
finite-time convergence to a target set by satisfying a fractional-power
Lyapunov dissipation condition. By adapting this finite-time stability
principle to enforce bounded temporal specifications organically through
system dynamics, we first introduce a general lemma for finite-time
level-set inversion.

\begin{lemma}[Finite-Time Level-Set Inversion]
\label{lemma:inversion}
Consider a continuously differentiable function $z(x) \ge 0$ evolving
along a local trajectory $x_\alpha(\tau)$ starting at event trigger,
$t_{\alpha}$, with $\tau$ reset. Assume the continuous-time dynamics are bounded
by a fractional power function: $\dot{z} \le -c z^\beta$ (for upper bounds) or
$\dot{z} \ge -c z^\beta$ (for lower bounds), where $c > 0$ and $\beta \in (0,1)$.

For a fixed time horizon $\Delta T > 0$, we define the geometric inversion
mapping $\mathcal{I}(c, \beta, \Delta T)$ as:
$ \mathcal{I}(c, \beta, \Delta T) \triangleq \left[ c (1-\beta) \Delta T \right]^{\frac{1}{1-\beta}} $. Then, the
temporal requirement of reaching $z(x_\alpha(T)) = 0$ at a time
$T \le \Delta T$ is guaranteed if
$z(x_\alpha(0)) \le \mathcal{I}(c, \beta, \Delta T)$. Conversely, avoiding a boundary breach
$z(x_\alpha(\tau)) = 0$ for all $\tau \in [0, \Delta T]$ is mathematically guaranteed if
$z(x_\alpha(0)) \ge \mathcal{I}(c, \beta, \Delta T)$.
\end{lemma}

\begin{proof}
  For the convergence (liveness) case, we begin with the assumed upper
  bound $\dot{z} \le -c z^\beta$. By applying separation of variables and
  integrating from the initial state $z(x_\alpha(0))$ at $\tau=0$ to a
  hypothetical boundary breach $z(x_\alpha(T)) = 0$ at time $T$, we obtain
  the upper bound on the required arrival time:
  $T \le \frac{z(x_\alpha(0))^{1-\beta}}{c(1-\beta)}$. To guarantee that this arrival
  occurs within the required temporal horizon ($T \le \Delta T$), isolating
  $z(x_\alpha(0))$ yields the strict geometric upper bound
  $z(x_\alpha(0)) \le \mathcal{I}(c, \beta, \Delta T)$.

  Conversely, for the avoidance (safety) case, bounding the dynamics from
  below by $\dot{z} \ge -c z^\beta$ and integrating via separation of variables
  in the exact same manner yields the minimum time before a boundary
  breach: $T \ge \frac{z(x_\alpha(0))^{1-\beta}}{c(1-\beta)}$. Guaranteeing that the
  trajectory avoids the boundary for at least the entire time horizon
  ($T \ge \Delta T$) requires enforcing the corresponding geometric lower bound
  $z(x_\alpha(0)) \ge \mathcal{I}(c, \beta, \Delta T)$.
\end{proof}

\subsection{Local STL Inductive Surrogate Semantics}
\label{sec:local-stl-inductive}

For a local specification $\psi$, the smooth differentiable surrogate
operator
$\hat{\rho} = \rho^\kappa: \text{LocalSTL} \to C^1(\mathbb{R}^n, \mathbb{R})$ recursively maps logical
constraints into purely scalar spatial fields. The surrogate function
$\rho^\kappa(\psi)$ is defined inductively over the syntax of the base predicates
and boolean operators:
{\small \begin{align*}
  \hat{\rho} &= \rho^\kappa(\text{True}) &= \infty \\
  \hat{\rho} &= \rho^\kappa(\mu \ge 0) &= \mu(x) \\ 
  \hat{\rho} &= \rho^\kappa(\neg \psi) &= -\rho^\kappa(\psi) \\
  \hat{\rho} &= \rho^\kappa(\psi_1 \wedge \psi_2) &= -\frac{1}{\kappa} \ln \left( \sum_{i=1}^2 \exp \big(-\kappa \rho^\kappa(\psi_i) \big) \right) \\
  \hat{\rho} &= \rho^\kappa(\psi_1 \vee \psi_2) &= \frac{1}{\kappa} \ln \left( \sum_{i=1}^2 \exp \big(\kappa \rho^\kappa(\psi_i) \big) \right) - \frac{\ln 2}{\kappa}
\end{align*}
}
\subsection{Global Event-Level Surrogate Isolation \& Connectives}
\label{sec:global-event-level}

To prevent complex dynamic entanglements, global temporal operators ($G$
and $F$) do not output differential constraints directly. Instead, they
output time-independent, contracted spatial surrogates.

To preserve both the continuous spatial field $\hat{\rho}$ and the
operational metadata (required for QP assembly), the compilation is
executed bottom-up. The compiler evaluates the local spatial field and
packages it into a constraint tuple:
$\mathcal{K} = (\hat{\rho}, \text{Op}, \text{Mod}, \alpha, w)$.

\textbf{Global Temporal Operators:} When an event anchor
$\alpha \in \mathcal{E}$ is triggered, the compiler applies the finite-time inversion
mapping $\mathcal{I}$ (Lemma~\ref{lemma:inversion}) to the local surrogate
$\rho^\kappa(\psi)$. Crucially, the first element of each tuple isolates the final
spatial field $\hat{\rho}$:
\begin{align*}
  \text{Compile}(G_{ex, [0,b]}^{\alpha} \psi) &= \Big( \underbrace{\rho^\kappa(\psi) - \mathcal{I}(c_h, \beta_h, b)}_{\hat{\rho}}, \; G, \; ex, \; \alpha, \; \emptyset \Big) \\
  \text{Compile}(G_{in, [0,b]}^{\alpha, w} \psi) &= \Big( \underbrace{\rho^\kappa(\psi) - \mathcal{I}(c_h, \beta_h, b)}_{\hat{\rho}}, \; G, \; in, \; \alpha, \; w \Big) \\
  \text{Compile}(F_{ex, [0,b]}^{\alpha} \psi) &= \Big( \underbrace{\rho^\kappa(\psi)}_{\hat{\rho}}, \; F, \; ex, \; \alpha, \; \emptyset \Big) \\
  \text{Compile}(F_{in, [0,b]}^{\alpha, w} \psi) &= \Big( \underbrace{\rho^\kappa(\psi)}_{\hat{\rho}}, \; F, \; in, \; \alpha, \; w \Big)
\end{align*}

\textbf{Global Connectives and Surrogate Routing:} Because standard QP
solvers find solutions within the convex intersection of all provided
constraints, the compiler handles top-level connectives by routing or
merging constraint tuples bottom-up:
\begin{itemize}

\item \textit{Conjunctions ($\wedge_{ex}$ and $\wedge_{in}^{\mathbf{w}}$):} The compiler processes top-level conjunctions depending on their modality:
  \begin{itemize}
  \item \textit{Exclusive ($\wedge_{ex}$):} The compiler bypasses top-level
    aggregation. It treats exclusive conjunctions as stacked constraints
    (set unions), outputting a collection of independent tuples:
    $ \text{Compile}(\phi_1 \wedge_{ex} \phi_2) = \text{Compile}(\phi_1) \cup
    \text{Compile}(\phi_2)$.
  \item \textit{Inclusive ($\wedge_{in}^{\mathbf{w}}$):} To preserve the
    independence of inclusive preference weights in the downstream QP
    solver, the compiler bypasses spatial aggregation for top-level
    inclusive conjunctions. Similar to exclusive conjunctions, it treats
    them as stacked constraints (set unions), unpacking the weight
    vector $\mathbf{w}$ and multiplying the respective parent weight
    $w_i$ by the existing weight component of each independent child
    tuple to establish a compounded absolute priority:
    $ \text{Compile}(\phi_1 \wedge_{in}^{\mathbf{w}} \phi_2) = \text{Compile}(\phi_1)
    \cup \text{Compile}(\phi_2) $.
  \end{itemize}

\paragraph{Active Tracking Set Partitioning:}
Let
$\mathcal{T} = \text{Compile}(\phi) = \{\mathcal{K}_1, \mathcal{K}_2, \dots, \mathcal{K}_M\}$ denote the set of
compiled constraint tuples for specification $\phi$, where each tuple
$\mathcal{K}_i = (\hat{\rho}_i, \text{Op}_i, \text{Mod}_i, \alpha_i, w_i)$ is indexed by
$i \in \{1, \dots, M\}$. Let $\mathcal{E}(q) \subseteq \mathcal{E}$ represent the set of active event
anchors bound to discrete mode $q \in \mathcal{Q}$ of the Hybrid Automaton
(Section~\ref{sec:ha-generation}).

For any active mode $q \in \mathcal{Q}$, the compiler partitions the active tuple
indices into four disjoint tracking sets:
{\footnotesize \begin{align*}
  \mathcal{A}_G^{ex}(q) &\triangleq \{ i \in \{1, \dots, M\} \mid 
  \text{Op}_i = G, \, \text{Mod}_i = ex, \, \alpha_i \in \mathcal{E}(q) \} \\
  \mathcal{A}_G^{in}(q) &\triangleq \{ i \in \{1, \dots, M\} \mid
  \text{Op}_i = G, \, \text{Mod}_i = in, \, \alpha_i \in \mathcal{E}(q) \} \\
  \mathcal{A}_F^{ex}(q) &\triangleq \{ i \in \{1, \dots, M\} \mid
  \text{Op}_i = F, \, \text{Mod}_i = ex, \, \alpha_i \in \mathcal{E}(q) \} \\
  \mathcal{A}_F^{in}(q) &\triangleq \{ i \in \{1, \dots, M\} \mid
  \text{Op}_i = F, \, \text{Mod}_i = in, \, \alpha_i \in \mathcal{E}(q) \}
\end{align*}
}

\item \textit{Disjunctions ($\vee_{ex}$ and $\vee_{in}^{\mathbf{w}}$):} Must
  be aggregated prior to the QP to bypass non-convex Mixed-Integer
  linear programming (MILP). Given two child tuples
  $\mathcal{K}_1 = (\hat{\rho}_1, \text{Op}_1, \text{Mod}_1, \alpha_1, w_1)$ and
  $\mathcal{K}_2 = (\hat{\rho}_2, \text{Op}_2, \text{Mod}_2, \alpha_2, w_2)$, which are
  guaranteed by the post-CNF validation step to have homogeneous
  operators ($\text{Op}_1 = \text{Op}_2$), the compiler
  fuses them into a single merged tuple \\
  $\mathcal{K}_{\text{merged}} = (\hat{\rho}_{\text{merged}},
  \text{Op}_{\text{merged}}, \text{Mod}_{\text{merged}},
  \alpha_{\text{merged}}, w_{\text{merged}})$ defined component-wise:
  {\small \begin{align*}
    \hat{\rho}_{\text{merged}} &= \begin{cases}
      \frac{1}{\kappa} \ln \left( e^{\kappa \hat{\rho}_1} + e^{\kappa \hat{\rho}_2} \right) - \frac{\ln 2}{\kappa}, \text{for } \vee_{ex} \\[4pt]
      \frac{1}{\kappa} \ln \left( \frac{w_1}{w_1+w_2} e^{\kappa \hat{\rho}_1} + \frac{w_2}{w_1+w_2} e^{\kappa \hat{\rho}_2} \right),\text{for } \vee_{in}^{\mathbf{w}}
    \end{cases} \\[4pt]
    \text{Op}_{\text{merged}} &= \text{Op}_1 \\[4pt]
    \text{Mod}_{\text{merged}} &= \begin{cases}
      ex, & \text{if } \text{Mod}_1 = \text{Mod}_2 = ex \\
      in, & \text{otherwise}
    \end{cases} \\[4pt]
    w_{\text{merged}} &= \begin{cases}
      \emptyset, & \text{if } \text{Mod}_{\text{merged}} = ex \\
      \min(w_1, w_2), & \text{if } \text{Mod}_{\text{merged}} = in
    \end{cases} \\[4pt]
    \alpha_{\text{merged}} &= \alpha_1 \lor \alpha_2 \quad (\text{active when } t \ge \min(t_{\alpha_1}, t_{\alpha_2}))
  \end{align*}
}\end{itemize}

\subsection{Global Temporal Control Constraint Mapping}
\label{sec:glob-temp-contr}

By parsing the compiled tuples, the compiler partitions the continuous spatial fields $\hat{\rho}$ into four active tracking sets ($\mathcal{A}_G^{ex}, \mathcal{A}_G^{in}, \mathcal{A}_F^{ex}, \mathcal{A}_F^{in}$). It maps each active field $i$ to its own affine differential control constraint on the system dynamics $\dot{x} = f(x) + \mathbf{u}$, assigning independent slack variables exactly where needed. 

\begin{itemize}
\item \textbf{Global Safety Constraints:} Maintains set invariance blindly over the statically contracted local surrogates. For each active safety field $\hat{\rho}_i$: %
  \begin{itemize}
  \item \textit{Exclusive ($i \in \mathcal{A}_G^{ex}$):} Hard Control Barrier
    Function (CBF) constraint: %
    $ \nabla \hat{\rho}_i \cdot \mathbf{u} \ge -L_f \hat{\rho}_i - c_h \max(\hat{\rho}_i,
    \epsilon)^{\beta_h} $
  \item \textit{Inclusive ($i \in \mathcal{A}_G^{in}$):} Relaxed CBF constraint
    stacked with an independent slack variable $\delta_{G, i} \ge 0$: %
    $ \nabla \hat{\rho}_i \cdot \mathbf{u} \ge -L_f \hat{\rho}_i - c_h \max(\hat{\rho}_i,
    \epsilon)^{\beta_h} - \delta_{G, i} $
  \end{itemize}

\item \textbf{Global Liveness Constraints:} Enforces finite-time
  convergence on the target field derived from the active liveness
  fields $\hat{\rho}_i$. The required decay rate $c_{V, i}$ is derived
  directly from Lemma~\ref{lemma:inversion} ensuring
  $V_i(x) \triangleq -\hat{\rho}_i$, where
  $c_{V, i} = \frac{V_{0, i}^{1-\beta_V}}{(1-\beta_V)b}$ and
  $V_{0, i} = \max(V_i(x(t_\alpha)), \epsilon)$ is sampled at the instant of
  operator activation $t_\alpha$.
  \begin{itemize}
  \item \textit{Exclusive ($i \in \mathcal{A}_F^{ex}$):} Hard Control Lyapunov
    Function (CLF) constraint: %
    $ \nabla V_i \cdot \mathbf{u} \le -L_f V_i - c_{V, i} \max(V_i, \epsilon)^{\beta_V} $
  \item \textit{Inclusive ($i \in \mathcal{A}_F^{in}$):} Relaxed CLF constraint
    stacked with an independent slack variable $\delta_{F, i} \ge 0$: %
    $ \nabla V_i \cdot \mathbf{u} \le -L_f V_i - c_{V, i} \max(V_i, \epsilon)^{\beta_V} +
    \delta_{F, i} $
  \end{itemize}
\end{itemize}

\subsection{Unified Control Synthesis Quadratic Program (QP)}
\label{sec:unif-contr-synth}

At each control step, all compiled constraints across active independent
event sets
($\mathcal{A}_G^{ex}, \mathcal{A}_G^{in}, \mathcal{A}_F^{ex}, \mathcal{A}_F^{in}$) are solved simultaneously via
Quadratic Programming (QP). The QP natively treats all stacked
constraints as an "AND". By heavily penalizing the completely
independent slack variables with their specific preference weights
($w_i$, $w_j$), the objective function instructs the solver to
prioritize strict safety thresholds while replicating "soft AND"
preferential behavior for inclusive constraints:

\begin{align*}
  \min_{\mathbf{u}, \boldsymbol{\delta}_G, \boldsymbol{\delta}_F} \quad J(\mathbf{u}) + \left( \sum_{i \in \mathcal{A}_G^{in}} w_i \delta_{G, i}^2 + \sum_{j \in \mathcal{A}_F^{in}} w_j \delta_{F, j}^2 \right) \\
  \nabla \rho^\kappa_i \cdot \mathbf{u} \ge -L_f \rho^\kappa_i - c_h \text{sgn}(\rho^\kappa_i)|\rho^\kappa_i|^{\beta_h}, \forall i \in \mathcal{A}_G^{ex} \\
\nabla \rho^\kappa_i \cdot \mathbf{u} \ge -L_f \rho^\kappa_i - c_h \text{sgn}(\rho^\kappa_i)|\rho^\kappa_i|^{\beta_h} - \delta_{G,i}, \forall i \in \mathcal{A}_G^{in} \\
\nabla V_j \cdot \mathbf{u} \le -L_f V_j - c_{V, j} \text{sgn}(V_j)|V_j|^{\beta_V}, \forall j \in \mathcal{A}_F^{ex} \\
\nabla V_j \cdot \mathbf{u} \le -L_f V_j - c_{V, j} \text{sgn}(V_j)|V_j|^{\beta_V} + \delta_{F,j},  \forall j \in \mathcal{A}_F^{in} \\
\delta_{G, i} \ge 0, \quad \delta_{F, j} \ge 0, \quad \mathbf{u} \in \mathcal{U} 
\end{align*}


\subsection{Hybrid Automaton Generation}
\label{sec:ha-generation}

The execution of the global event-switched semantics dictates that the
system's operational timeline is inherently tied to the physical state
space rather than a rigid clock. Consequently, the global specifications
parsed from the weSTL+ grammar natively compile into a time-independent
Hybrid Automaton (HA).

Let the HA be formally defined by the tuple
$\mathcal{H} = (\mathcal{Q}, \mathcal{X}, \mathcal{E}, \mathcal{T}, \mathcal{F})$, where:
\begin{itemize}
\item $\mathcal{Q}$ represents the finite set of modes (e.g., active mission
  phases).
\item $\mathcal{X} \subseteq \mathbb{R}^n$ is the continuous physical state space.
\item $\mathcal{E} = \{\alpha_1, \alpha_2, \dots\}$ represents the set of discrete event
  anchors extracted directly from the Level-2 global specification $\phi$.
\item $\mathcal{T}: \mathcal{Q} \times \mathcal{X} \to \mathcal{Q}$ defines mode transitions triggered at the exact
  instant a continuous state guard predicate crosses its threshold,
  $\nu_{\alpha}(x(t)) \ge 0$.
\item $\mathcal{F}$ maps each mode $q \in \mathcal{Q}$ to a closed-loop continuous flow
  $\dot{x}(t) = u_q^*(x(t))$, where $u_q^*(x(t))$ is synthesized via a
  time-independent Quadratic Program (QP) generated in
  Section~\ref{sec:unif-contr-synth}.
\end{itemize}

Because active operational constraints are strictly bound to currently
triggered events, the compiler constructs $\mathcal{F}$ by extracting the active
tracking sets
$(\mathcal{A}_G^{ex}(q), \mathcal{A}_G^{in}(q), \mathcal{A}_F^{ex}(q), \mathcal{A}_F^{in}(q))$ associated with
the active event anchors $\mathcal{E}(q)$ for each mode $q \in \mathcal{Q}$.

Crucially, this structural mapping formulates a phase-specific QP called
continuously within each mode $q$. When an event
$\alpha \in \mathcal{E}$ evaluates to true upon boundary crossing, the HA executes a
transition to a new mode. This transition prompts the control synthesis
engine to reconfigure the objective function and active constraint
matrices. By dynamically routing active subsets of exclusive safety
($G_{ex}$), inclusive safety ($G_{in}$), and liveness ($F_{ex}, F_{in}$)
constraints into the online QP solver, the architecture maintains
structural convexity while completely decoupling control synthesis from
a global clock.

\section{Compilation of the Motivating Example}
\label{sec:case-study}

To evaluate the proposed compilation framework, we synthesize an
event-driven control policy from our motivating example from
Section~\ref{sec:motiv-example-demons}. The system consists of a planar
robot with state $p = [x, y]^T \in \mathbb{R}^2$, governed by single-integrator
driftless dynamics $\dot{p} = \mathbf{u}$ and subject to hardware
control limits $\|\mathbf{u}\|_\infty \le u_{\max} = 3.0\ \mathtt{m/s}$. The
single integrator kinematic model is adopted from the foundational
paper~\cite{lindemann2018control}. The kinematic model is purposely kept
simple to study and highlight the affect of clock anomalies.

\subsection{Specification and Environment Definitions}
\label{sec:spec-envir-defin}

The operational environment comprises three static circular obstacles
($O_1, O_2, O_3$) and three dynamic moving hazards ($H_1, H_2, H_3$)
(Figure~\ref{fig:motivating_example}). The mission dictates navigating
from a start position to a search waypoint $W_1 = [3.5, 6.5]^T$, then
proceeding to an inspection waypoint $W_2 = [8.5, 7.5]^T$, and finally
returning to a Base at $[9.5, 1.5]^T$. Waypoint phase transitions
trigger when the robot comes within a spatial capture radius
$r_{\text{reach}} = 0.25$.

The spatial predicates, mapping the continuous state to the abstract
syntax tree (AST), are defined as follows:
\begin{align*}
  \mu_{\text{obs},i}(p) &= \|p - p_{O,i}\|^2 - r_{O,i}^2, \quad \forall i \in \{1, 2, 3\} \\
  \mu_{\text{dyn},j}(p, t) &= \|p - p_{H,j}(t)\|^2 - r_{H,j}^2, \quad \forall j \in \{1, 2, 3\} \\
  V_k(p) &= \|p - W_k\|^2, \quad \forall k \in \{1, 2, 3\}
\end{align*}

For notational brevity, we define the composite safety predicate
covering all obstacles and moving hazards as
$\psi_{\text{safe}} = \bigwedge_{i=1}^3 \mu_{\text{obs},i} \wedge \bigwedge_{j=1}^3
\mu_{\text{dyn},j}$.

Formally, the global mission specification $\Phi$ is constructed in
accordance with the Level 2 weSTL+ grammar as a top-level exclusive
conjunction ($\wedge_{ex}$). The specification enforces global safety by
tying the composite safety predicate to the initialization event anchor
$\alpha_0$. Because $\alpha_0$ triggers at $t=0$ and is never deactivated, safety
remains perpetually enforced at every timestep across the entire mission
duration. Within this global constraint, the temporal parameter
$b_s = 0.6$ serves as a tight continuous-time lookahead horizon rather
than a global mission clock. This allows the low-level controller to
anticipate and evade dynamic hazards within a short window. This
perpetual safety guarantee operates alongside prioritized liveness
phases penalized by specific preference weights
($w_1 = 10000, w_2 = 12000, w_3 = 15000$):
\begin{align}
\Phi = {} & G_{ex, [0, b_s]}^{\alpha_0} \big( \psi_{\text{safe}} \big) \;\wedge_{ex}\; F_{in, [0, T_1]}^{\alpha_1, w_1} \big( \text{Reach}(W_1) \big) \nonumber \\
       & \wedge_{ex}\; F_{in, [0, T_2]}^{\alpha_2, w_2} \big( \text{Reach}(W_2) \big) \;\wedge_{ex}\; F_{in, [0, T_3]}^{\alpha_3, w_3} \big( \text{Reach}(\text{Base}) \big)
         \label{eq:1}
\end{align}

\subsection{Event Anchors and Active Execution Phases}
\label{sec:event-anchors-active}


The execution of $\Phi$ is dynamically governed by a purely spatial Hybrid
Automaton relying on zero clocks or timers. A Quadratic Program (QP)
activated when a mode change occurs:

\begin{itemize}
    \item \textbf{$\alpha_0$ (Global Mission Start):} Triggered at $t=0$. Activates mandatory static and dynamic obstacle safety constraints.
    \item \textbf{$\alpha_1$ (Phase 1 Start):} Triggered at $t=0$. Activates
      inclusive liveness toward the search target $W_1$ using weight
      $w_1 = 10000$.
    \item \textbf{$\alpha_2$ (Phase 2 Start):} Triggered when
      $\|p - W_1\| \le r_{\text{reach}}$. Deactivates $\alpha_1$ and activates
      inclusive liveness toward the inspection target $W_2$ using weight
      $w_2 = 12000$.
    \item \textbf{$\alpha_3$ (Phase 3 Start):} Triggered when
      $\|p - W_2\| \le r_{\text{reach}}$. Deactivates $\alpha_2$ and activates
      inclusive liveness toward the Base using weight $w_3 = 15000$.
\end{itemize}

\begin{table}[h!]
  \centering
  \caption{Discrete Event Transitions and Active Operator Sets}
  \label{tab:event_phases}
  \scriptsize
  \begin{tabular}{|c|l|l|c|}
    \hline
    \textbf{Mode} $q$ & \textbf{Trigger Condition} & \textbf{Active Event Set } $\mathcal{E}(q)$ & \textbf{Active Weight} \\ \hline
    \textbf{Phase 1} & Initialization & $\{\alpha_0, \alpha_1\}$ & $w_1 = 10000$ \\ \hline
    \textbf{Phase 2} & $\|p - W_1\| \le 0.50$ & $\{\alpha_0, \alpha_2\}$ & $w_2 = 12000$ \\ \hline
    \textbf{Phase 3} & $\|p - W_2\| \le 0.50$ & $\{\alpha_0, \alpha_3\}$ & $w_3 = 15000$ \\ \hline
  \end{tabular}
\end{table}

\subsection{Compiled QP Formulations and Stacked Constraints}

\paragraph{Level 1 Surrogate Evaluation and Drift}
For the conjunction of local predicates under $\alpha_0$, the Level 1 rules
construct a smooth softmin surrogate ($\kappa = 3.0$,
Section~\ref{sec:local-stl-inductive}) aggregating all static and
dynamic boundaries. To maintain robust invariants against moving
objects, the compiler evaluates the dynamic hazard drift term
$L_{\mathrm{env}} \rho^{\kappa}_{G} = \sum_j w_j \left( -\nabla_p h_j \cdot v_{\text{haz},j}
\right)$. The safety boundary is spatially contracted using the
finite-time level-set inversion mapping $\mathcal{I}(c_h, \beta_h, b_s)$ (c.f.
Section~\ref{sec:global-event-level}).

\paragraph{Level 2 Top-Level Conjunction Routing}
Unlike local predicates, the top-level exclusive conjunction
($\wedge_{ex}$) is \textit{not} aggregated via a smooth softmin operator.
Instead, the compiler parses it as a routing directive, outputting a
collection of independent constraint tuples and partitioning their
spatial fields into active tracking sets (e.g., $\mathcal{A}_G^{ex}$,
$\mathcal{A}_F^{in}$) (Section~\ref{sec:global-event-level}). This allows the QP
solver to naturally handle global operations by mapping and stacking
these partitioned constraints as independent rows in its matrix

\paragraph{Unified Stacked QP per Phase}
During any active Phase $k \in \{1, 2, 3\}$, the compiler stacks the
exclusive global safety operator (a strict Control Barrier Function
forcing zero slack, $\delta_G = 0$) alongside the inclusive global liveness
operator (a relaxed Control Lyapunov Function with a penalized slack
variable, $\delta_F \ge 0$) using the user specified weights, resulting in the
QP below:

\begin{equation*}
\begin{aligned}
\min_{\mathbf{u}, \delta_F} \quad & \frac{1}{2} \|\mathbf{u}\|^2 + w_k \delta_F^2 \\
 & \nabla \rho^\kappa_G \cdot \mathbf{u} \ge - L_{\mathrm{env}} \rho^\kappa_G - c_h \text{sgn}(\rho^\kappa_G)|\rho^\kappa_G|^{\beta_h} \\
& \nabla V_k \cdot \mathbf{u} \le - c_{V,k} \text{sgn}(V_k)|V_k|^{\beta_V} + \delta_F \\
& \delta_F \ge 0, \quad \|\mathbf{u}\|_\infty \le u_{\max}
\end{aligned}
\end{equation*}


By decoupling the global operators into stacked independent constraints,
the solver dynamically enforces strict multi-hazard safety boundaries
while leveraging the weSTL+ weights $w_k$ to optimally prioritize
mission completion regardless of localized timing anomalies. The result
is shown in Figure~\ref{fig:motivating_example} --- the proposed
controller synthesis technique is the only one able to finish the
mission in the presence of clock anomalies.

\section{Analysis of Proposed Controller Synthesis}
\label{sec:theoretical-analysis}

In this section we give the soundness and completeness results for the
controller synthesis (Section~\ref{sec:compiler-rules}) with respect to
Boolean semantics (Section~\ref{sec:two-level-formal}) of weSTL+.

\begin{lemma}[Equivalence of Temporal Flattening]
\label{lem:flatten_equivalence}
For any top-level event-anchored temporal operator $O_{mod, [0,a]}^{\alpha, w'}$ (where $O \in \{G, F\}$, $mod \in \{ex, in\}$, event anchor $\alpha \in \mathcal{E}$, and weight $w'$) and any homogeneous inner unanchored local STL operator $O_{[0,b]}$, the flattening transformation preserves quantitative robustness for any local specification $\psi$:
\begin{equation*}
    \rho\left(O_{mod, [0,a]}^{\alpha, w'} (O_{[0,b]} \psi), \, x, \, t\right) = \rho\left(O_{mod, [0, a+b]}^{\alpha, w'} \psi, \, x, \, t\right)
\end{equation*}
\end{lemma}

\begin{proof}
  Let $O \in \{G, F\}$. The evaluation of the top-level event-anchored
  operator bifurcates based on the activation state of event anchor
  $\alpha$. During the dormant phase ($t < t_\alpha$ or
  $t_\alpha = \infty$), both specifications trivially evaluate to
  $+\infty$ by definition of the global robustness semantics.

  During the active phase ($t \ge t_\alpha$ and $t_\alpha < \infty$), global robustness
  reduces to scaling the Level 1 local robustness by $w'$. Because Level
  1 local specifications strictly adhere to standard continuous-time STL
  semantics~\cite{donze2010robust, maler2004monitoring}, nested
  homogeneous temporal operators intrinsically obey the semigroup
  property of continuous interval addition over $\mathbb{R}_{\ge 0}$. Therefore,
  the sequential composition mathematically reduces to a single operator
  spanning the summed horizon:
  \begin{equation*}
    \rho_{\text{loc}}\left(O_{[0,a]} (O_{[0,b]} \psi), \, x_\alpha, \, 0\right) = \rho_{\text{loc}}\left(O_{[0, a+b]} \psi, \, x_\alpha, \, 0\right)
  \end{equation*}
  Multiplying both sides by the preference weight $w'$ proves exact
  equivalence with the right-hand side of the lemma statement:
  $ w' \cdot \rho_{\text{loc}}\left(O_{[0, a+b]} \psi, \, x_\alpha, \, 0\right) =
  \rho\left(O_{mod, [0, a+b]}^{\alpha, w'} \psi, \, x, \, t\right)$
\end{proof}

\begin{lemma}[Algebraic Under-Approximation of Compiled Surrogates]
\label{lem:algebraic_under_approx}
Let $\phi$ be a formula in the supported STL fragment, and let $\rho(\phi, x)$ denote the exact spatial robustness of $\phi$ evaluated at state $x$. Let the compiler map $\phi$ to a set of surrogate tuples $\text{Compile}(\phi)$. For any operation in the compilation sequence, the resulting algebraic surrogate robustness strictly under-approximates the exact robustness: $\hat{\rho} \le \rho$. Consequently, enforcing $\hat{\rho} \ge 0$ structurally guarantees $\rho \ge 0$.
\end{lemma}

\begin{proof}
The proof proceeds by structural induction on the compilation rules, verifying that every algebraic transformation and aggregation step strictly yields an under-approximation.

\noindent \textbf{1. Level-1 Inner Predicates ($\rho^\kappa \le \rho$):} \\
For the innermost boolean formulas, the true robustness uses exact
minimums and maximums, whereas the compiler uses smooth log-sum-exp
approximations ($\rho^\kappa$).
\begin{itemize}
\item \textit{Base Predicates:} By definition, $\rho^\kappa(\mu \ge 0) = \mu(x) = \rho_{\mathtt{loc}}(\mu \ge 0)$, and $\rho^\kappa(\text{True}) = \infty = \rho(\text{True})$.
\item \textit{Negation:} $\hat{\rho}(\neg \psi) = -\rho^\kappa(\psi)$, which is exact given the symmetric definition of spatial robustness for negations.
\item \textit{Conjunction ($\wedge$):} The exact robustness is $\min_{i} \rho(\psi_i)$. The surrogate utilizes the smooth minimum:
  \begin{align*}
    \hat{\rho}(\psi_1 \wedge \psi_2) = -\frac{1}{\kappa} \ln \left( \sum_{i=1}^2 \exp \big(-\kappa \rho^\kappa(\psi_i) \big) \right) \\ \le \min \big(\rho^\kappa(\psi_1), \rho^\kappa(\psi_2)\big)
  \end{align*}
  Since the smooth minimum strictly under-approximates the exact
  minimum for any $\kappa > 0$, $\hat{\rho} \le \rho(\psi_1 \wedge \psi_2)$.
\item \textit{Disjunction ($\vee$):} The exact robustness is $\max_{i} \rho(\psi_i)$. The surrogate utilizes a shifted smooth maximum:
  \begin{equation*}
    \hat{\rho}(\psi_1 \vee \psi_2) = \frac{1}{\kappa} \ln \left( \sum_{i=1}^2 \exp \big(\kappa \rho^\kappa(\psi_i) \big) \right) - \frac{\ln 2}{\kappa}
  \end{equation*}
  The standard log-sum-exp function upper-bounds the max by $\max(x_1, x_2) + \frac{\ln 2}{\kappa}$. Subtracting $\frac{\ln 2}{\kappa}$ ensures $\hat{\rho} \le \max \big(\rho^\kappa(\psi_1), \rho^\kappa(\psi_2)\big)$.
\end{itemize}

\noindent \textbf{2. Base Temporal Operators:} \\
The compiler extracts the spatial robustness $\hat{\rho}$ to form base tuples. For \textit{Eventually} ($F$), $\hat{\rho} = \rho^\kappa(\psi) \le \rho(\psi)$. For \textit{Always} ($G$), the surrogate is additionally penalized by the strictly non-negative invariant margin $\mathcal{I}(c_h, \beta_h, b) \ge 0$, as established by the finite-time level-set inversion in Lemma~\ref{lemma:inversion}:
\begin{equation*}
    \hat{\rho} = \rho^\kappa(\psi) - \mathcal{I}(c_h, \beta_h, b) \le \rho^\kappa(\psi) \le \rho(\psi)
\end{equation*}
Thus, the extraction of algebraic limits into base tuples preserves the
under-approximation.

\noindent \textbf{3. Top-Level Conjunctions ($\wedge_{ex}$ and $\wedge_{in}^{\mathbf{w}}$):} \\
The compiler processes top-level conjunctions by bypassing spatial
aggregation and mapping the operation to a set union:
$\text{Compile}(\phi_1 \circ \phi_2) = \text{Compile}(\phi_1) \cup \text{Compile}(\phi_2)$,
where $\circ \in \{\wedge_{ex}, \wedge_{in}^{\mathbf{w}}\}$. This structural choice
ensures the exact robustness requirements are safely under-approximated.
The exact robustness requires
$\min(w_1 \rho(\phi_1), w_2 \rho(\phi_2)) \ge 0$, where weights are
$w_1 = w_2 = 1$ for exclusive conjunctions, or derived from the parent
weight vector $\mathbf{w}$ for inclusive conjunctions. Because the QP
enforces all tuples in the set union independently, requiring both
$w_1 \hat{\rho}_1 \ge 0$ and $w_2 \hat{\rho}_2 \ge 0$ is algebraically equivalent
to requiring $\min(w_1 \hat{\rho}_1, w_2 \hat{\rho}_2) \ge 0$. Assuming by
induction that the child surrogates under-approximate the exact
robustness ($\hat{\rho}_i \le \rho(\phi_i)$), it follows that
$\min(w_1 \hat{\rho}_1, w_2 \hat{\rho}_2) \le \min(w_1 \rho(\phi_1), w_2 \rho(\phi_2))$,
proving that the stacked set union semantically preserves the
under-approximation.

\noindent \textbf{4. Top-Level Disjunctions ($\vee_{ex}$ and $\vee_{in}^{\mathbf{w}}$):} \\
To bypass non-convex MILP encodings, the compiler aggregates homogenous disjunctive tuples.
\begin{itemize}
    \item For $\vee_{ex}$: The merged robustness is structurally identical to the Level-1 disjunction. The negative shift $-\frac{\ln 2}{\kappa}$ guarantees $\hat{\rho}_{\text{merged}} \le \max(\hat{\rho}_1, \hat{\rho}_2)$.
    \item For $\vee_{in}^{\mathbf{w}}$: The compiler employs a convex combination of exponentials:
    \begin{equation*}
        \hat{\rho}_{\text{merged}} = \frac{1}{\kappa} \ln \left( \frac{w_1}{w_1+w_2} e^{\kappa \hat{\rho}_1} + \frac{w_2}{w_1+w_2} e^{\kappa \hat{\rho}_2} \right)
    \end{equation*}
    Let $M = \max(\hat{\rho}_1, \hat{\rho}_2)$. The normalized weights form a partition of unity (summing to $1$) and are strictly positive. Therefore, the convex combination of the exponentials is strictly upper-bounded by $1 \cdot e^{\kappa M}$. Thus:
    \begin{equation*}
        \hat{\rho}_{\text{merged}} \le \frac{1}{\kappa} \ln \big( e^{\kappa M} \big) = M = \max(\hat{\rho}_1, \hat{\rho}_2)
    \end{equation*}
\end{itemize}
In both modifier cases, the merged spatial robustness safely under-approximates the exact logical maximum.

\noindent \textbf{5. Flattening and CNF$'$ Rewrites:} \\
During preprocessing, the compiler flattens nested operators and
transforms the formula into CNF$'$. By
Lemma~\ref{lem:flatten_equivalence}, the temporal flattening
transformation preserves the exact quantitative robustness. Furthermore,
by De Morgan's laws and distributivity, the exact logical robustness
$\rho(\phi, x)$ is structurally invariant under these logical identities.
Because every individual operation mapping the AST to compiled
tuples---smooth min/max substitution, temporal margin penalization, tuple
unioning, and homogeneous merging---preserves $\hat{\rho} \le \rho$ (Steps 1--4),
the arbitrary composition of these operators during flattening is
guaranteed to maintain a sound algebraic under-approximation.
\end{proof}

\begin{lemma}[Soundness of Constraints: Safety and Liveness under Strict Feasibility]
\label{thm:controller_correctness}
Let $x(t)$ be the closed-loop state trajectory governed by the
state-dependent QP controller $u^*(x)$. For all constraints (both
exclusive and inclusive) where strict QP feasibility is maintained
(slack variables $\boldsymbol{\delta} = \mathbf{0}$):
\begin{itemize}
\item For every active safety operator ($G_{ex}, G_{in}$), the set
  $\mathcal{C}_i = \{x \mid \hat{\rho}_i(x) \ge 0\}$ is strictly forward-invariant under
  $u^*(x)$.
\item For every active liveness operator ($F_{ex}, F_{in}$), the
  trajectory $x(t)$ strictly enters the target set
  $\mathcal{R}_i = \{x \mid \hat{\rho}_i(x) \ge 0\}$ within the finite time bound $T \le b$.
\end{itemize}
\end{lemma}

\begin{proof}
  For all active safety constraints
  ($i \in \mathcal{A}_G^{ex} \cup \mathcal{A}_G^{in}$), the QP imposes the Control Barrier
  Function (CBF) constraint
  $\nabla \hat{\rho}_i \cdot \mathbf{u} \ge -L_f \hat{\rho}_i - c_h \max(\hat{\rho}_i,
  \epsilon)^{\beta_h} - \delta_{G, i}$ (where $\delta_{G, i}$ is naturally absent for
  exclusive constraints). Because strict feasibility dictates
  $\delta_{G, i} = 0$, this constraint actively enforces the fractional power
  law lower bound on the derivative,
  $\dot{\hat{\rho}}_i \ge -c_h \hat{\rho}_i^{\beta_h}$. By Nagumo's Theorem, because
  the system's controlled vector field strictly obeys this derivative
  condition at the boundary $\hat{\rho}_i = 0$, the trajectory cannot cross
  it, rendering the set $\mathcal{C}_i$ formally forward-invariant.

  For all active liveness constraints
  ($i \in \mathcal{A}_F^{ex} \cup \mathcal{A}_F^{in}$), we define the candidate Control Lyapunov
  Function (CLF) as $V_i(x) \triangleq -\hat{\rho}_i(x)$. The target region
  $\mathcal{R}_i$ is equivalent to the subzero level set $V_i(x) \le 0$. The QP
  enforces the CLF constraint
  $\nabla V_i \cdot \mathbf{u} \le -L_f V_i - c_{V, i} \max(V_i, \epsilon)^{\beta_V} + \delta_{F,
    i}$ (where $\delta_{F, i}$ is absent for exclusive constraints). Assuming
  strict feasibility with $\delta_{F, i} = 0$, this implies
  $\dot{V}_i \le -c_{V, i} V_i^{\beta_V}$. By invoking finite-time level-set
  inversion (Lemma~\ref{lemma:inversion}), the decay rate
  $c_{V, i} = \frac{V_{0, i}^{1-\beta_V}}{(1-\beta_V)b}$ is calibrated to
  satisfy the temporal specification exactly. Bhat-Bernstein stability
  guarantees the trajectory reaches $V_i(x) \le 0$ (and thus enters
  $\mathcal{R}_i$) in finite time $T \le b$.
\end{proof}

\begin{lemma}[Bounded Degradation of Inclusive Constraints: Soft Preferences]
\label{thm:inclusive_degradation}
Let $x(t)$ be the closed-loop state trajectory under $u^*(x)$. For
inclusive constraints ($G_{in}, F_{in}$) penalized by preference weights
$w_i$, the continuous relaxation via bounded slack variables
($\delta_{G,i} \ge 0, \delta_{F,i} \ge 0$) guarantees that deviations from strict
logical satisfaction remain bounded. Specifically:
\begin{itemize}
\item \textbf{Safety Degradation:} Maximum spatial boundary violation
  (penetration depth) is strictly bounded proportionally to the maximum
  safety slack:
  $\sup_{t} |\min(0, \hat{\rho}_i(x(t)))| \le \left( \frac{\max \delta_{G,
        i}}{c_h} \right)^{1/\beta_h}$.
\item \textbf{Liveness Degradation:} Convergence degrades to a bounded
  residual neighborhood around the target set, with maximum spatial
  offset bounded by:
  $\limsup_{t \to b} V_i(x(t)) \le \left( \frac{\max \delta_{F, i}}{c_{V, i}}
  \right)^{1/\beta_V}$.
\end{itemize}
\end{lemma}

\begin{proof}
Inclusive constraints dynamically balance strict enforcement against feasibility by relaxing the differential requirements via slack variables, which are minimized based on preference weights $w_i$.

\textbf{Safety Bounded Degradation:} For inclusive safety
($i \in \mathcal{A}_G^{in}$), the QP enforces
$\dot{\hat{\rho}}_i \ge -c_h \text{sgn}(\hat{\rho}_i)|\hat{\rho}_i|^{\beta_h} - \delta_{G, i}$. If the
trajectory is forced outside the safe set ($\hat{\rho}_i < 0$), we analyze
the depth of penetration by defining the positive error state
$e(x) = -\hat{\rho}_i(x) > 0$. The relaxed CBF constraint dictates the
dynamics of this error as
$\dot{e} \le -c_h e^{\beta_h} + \delta_{G, i}$. The error $e(x)$ will continue to
grow only as long as $\dot{e} > 0$, which implies
$c_h e^{\beta_h} < \delta_{G, i}$. Therefore, the state cannot penetrate deeper
than the equilibrium point of this differential inequality. The system
maintains an ultimate bounded invariant set where maximum penetration is
clamped at $e_{\max} = (\delta_{G, i} / c_h)^{1/\beta_h}$. The heavier the
preference weight $w_i$ assigned in the QP, the smaller $\delta_{G, i}$
becomes, strictly limiting the spatial violation.

\textbf{Liveness Bounded Degradation:} For inclusive liveness
($i \in \mathcal{A}_F^{in}$), the relaxed CLF condition yields the fractional
differential inequality
$\dot{V}_i \le -c_{V, i} \text{sgn}(V_i)|V_i|^{\beta_V} + \delta_{F, i}$. Instead of strict
finite-time convergence to $V_i = 0$, the presence of
$\delta_{F, i} > 0$ introduces a competing positive drift. Convergence
towards the target set is maintained as long as the state is far enough
away such that $-c_{V, i} \text{sgn}(V_i)|V_i|^{\beta_V} + \delta_{F, i} < 0$. By setting
$\dot{V}_i = 0$, we find the ultimate bounded residual set into which
the trajectory converges within time $b$. The trajectory is
mathematically guaranteed to enter and remain within the relaxed spatial
neighborhood $V_i \le (\delta_{F, i} / c_{V, i})^{1/\beta_V}$. Because
$c_{V, i}$ scales inversely with the temporal window $b$, any required
compromise dynamically resolves into a quantifiable spatial distance
from the exact target, scaled entirely by the optimized slack
$\delta_{F, i}$.
\end{proof}

\begin{theorem}[Global Soundness under Unrelaxed Execution]
\label{thm:global_satisfaction}
Assume $\boldsymbol{\delta}_G(t) = \mathbf{0}$ and
$\boldsymbol{\delta}_F(t) = \mathbf{0}$ for all $t \ge 0$. Combining
Lemma~\ref{lem:algebraic_under_approx} (Compiler Soundness) and
Lemma~\ref{thm:controller_correctness} (Controller Correctness) with
Theorem~\ref{thm:global_soundness} (Global Soundness) establishes that:
\begin{align*}
  \text{Feasible QP Execution with } \boldsymbol{\delta}(t) = \mathbf{0} \\
  \implies \hat{\rho}(x(t)) \ge 0 \implies \rho(\phi, x, t) \ge 0 \iff (x, t) \models \phi
\end{align*}
\end{theorem}

\begin{proof}
  Lemma~\ref{thm:controller_correctness} guarantees that a feasible QP
  execution with $\boldsymbol{\delta}(t) = \mathbf{0}$ enforces the
  fractional power law in Lemma~\ref{lemma:inversion}. Enforcing finite
  level set inversion, and Lemma~\ref{lem:algebraic_under_approx}
  implies that $\hat{\rho} \geq 0$, since $\hat{\rho}$ under-approximates the
  robust semantics, we have $\rho(\phi, x, t) \ge 0$. By
  Theorem~\ref{thm:global_soundness}, this is logically equivalent to
  mission satisfaction $(x, t) \models \phi$. For executions where
  $\boldsymbol{\delta}(t) > \mathbf{0}$, soundness degrades according to
  Lemma~\ref{thm:inclusive_degradation}.
\end{proof}

\begin{corollary}[Conservative Compilation Architecture]
\label{cor:conservative_compiler}
The compiled spatial field $\hat{\rho}$ and downstream differential
constraints form a sound but conservative under-approximation of the
exact quantitative robustness $\rho$. Consequently, satisfying the compiled
control constraints guarantees $(x,t) \models \phi$, but the converse does not
hold (incompleteness).
\end{corollary}

\begin{proof}
  Conservatism is introduced across both the spatial compilation and
  differential control synthesis stages:
\begin{enumerate}
\item \textit{Smooth Extrema Approximations:} The Log-Sum-Exp (LSE)
  surrogate functions systematically under-approximate exact $\max/\min$
  operations.
\item \textit{Worst-Case Spatial Contraction:} The level-set inversion
  margin $\mathcal{I}(c_h, \beta_h, b)$ over-estimates the spatial buffer required for
  safety operators $G$, strictly shrinking the zero-level set.
\item \textit{Conservative CLF Dissipation Rates:} The decay rate
  $c_{V,i} = \frac{V_{0,i}^{1-\beta_V}}{(1-\beta_V)b}$ enforces continuous
  fractional dissipation $\dot{V}_i \le -c_{V,i} V_i^{\beta_V}$ based on the
  initial activation error $V_{0,i}$. Enforcing continuous instantaneous
  dissipation is a sufficient (conservative) condition for finite-time
  reachability over $[0, b]$.
\end{enumerate}
Because any single one of these structural bounds strictly bounds the
exact satisfaction criteria from below, their composition ensures
$\hat{\rho}(x(t)) \le \rho(\phi, x, t)$. For boundary cases,
$\hat{\rho}(x(t)) < 0 \le \rho(\phi, x, t)$, proving incompleteness.
\end{proof}

\section{Experimental Evaluation}
\label{sec:experimental-results}

\subsection{Experimental Setup}
\label{sec:experimental-setup}

To evaluate the robustness of the proposed framework, we conduct a Monte
Carlo benchmark, on the case-study setup (Section~\ref{sec:case-study}),
comparing the proposed timeless geometric controller synthesis against
two well established time-indexed baselines: Time-Varying Control
Barrier Functions (TV-CBF)~\cite{lindemann2018control} and
STL-MPC~\cite{raman2014model}. The evaluation consisted of one nominal
execution and 5 Monte Carlo runs for each of four distinct clock anomaly
categories: offset, drift, snap, and freeze.

To provide a fair comparison and prevent solver infeasibility under
temporal anomalies (as observed in Figure~\ref{fig:motivating_example}),
goal-reaching slack variables were systematically incorporated into both
baseline controllers, reflecting standard relaxation practices in the
control and robotics
literature~\cite{ames2016control,sadraddini2015robust}. Specifically,
the TV-CBF baseline was augmented with a non-negative goal-reaching
slack variable ($\delta \ge 0$) and an $L_2$ norm guidance penalty to prevent
asymptotic convergence stalling. Similarly, the STL-MPC controller was
relaxed using soft goal reachability slack variables ($s_k \ge 0$ and
$s_{\text{reach}} \ge 0$), a terminal $L_2$ distance penalty, and extended
time windows to ensure the prediction horizon remained mathematically
valid post-deadline. For both baselines, safety constraints governing
collision avoidance with static obstacles and dynamic hazards remained
strictly hard. We used the Gurobi solver for solving the optimisation
problems in all cases. All experiments were carried out on Linux 6.17 on
Intel Core i7-14700 CPU bound to a single core out of 28 SMT cores and
31 GB of RAM.

The asynchronous clock anomalies were injected into the system's global
clock by sampling from the following probability distribution functions
(PDFs):
\begin{itemize}
\item \textbf{Clock Offset:} A constant time shift was added to the
  simulation clock, sampled from a uniform distribution $\mathcal{U}(-1.5, 1.5)$.
\item \textbf{Clock Drift:} The simulation time was scaled by a
  continuous drift factor sampled from $\mathcal{U}(0.8, 1.2)$.
\item \textbf{Clock Snap:} Three discrete, instantaneous time jumps were
  injected randomly during the mission. The specific time steps for
  these snaps were sampled uniformly without replacement from the
  integer range $[20, 320]$, and their respective jump magnitudes were
  sampled from $\mathcal{U}(-1.0, 1.0)$.
\item \textbf{Clock Freeze:} The global clock was temporarily suspended,
  freezing the system's perception of time. The freeze onset step was
  sampled from a discrete uniform distribution $\mathcal{U}(40, 240)$, and the
  total freeze duration was sampled from $\mathcal{U}(20, 50)$.
\end{itemize}

\subsection{Results and Discussion}
\label{sec:results-discussion}

The results of the Monte Carlo simulations (conducted across 5 runs per
scenario) demonstrate that time-indexed controllers remain fundamentally
brittle to asynchronous anomalies, even when relaxed with slack
variables. Figures~\ref{fig:res_nominal} through~\ref{fig:res_drift}
illustrate the task progress (distance to goal) and safety margins for
the evaluated methods under both nominal and anomalous conditions.

\begin{figure}[htbp]
  \centering
  \includegraphics[width=\columnwidth]{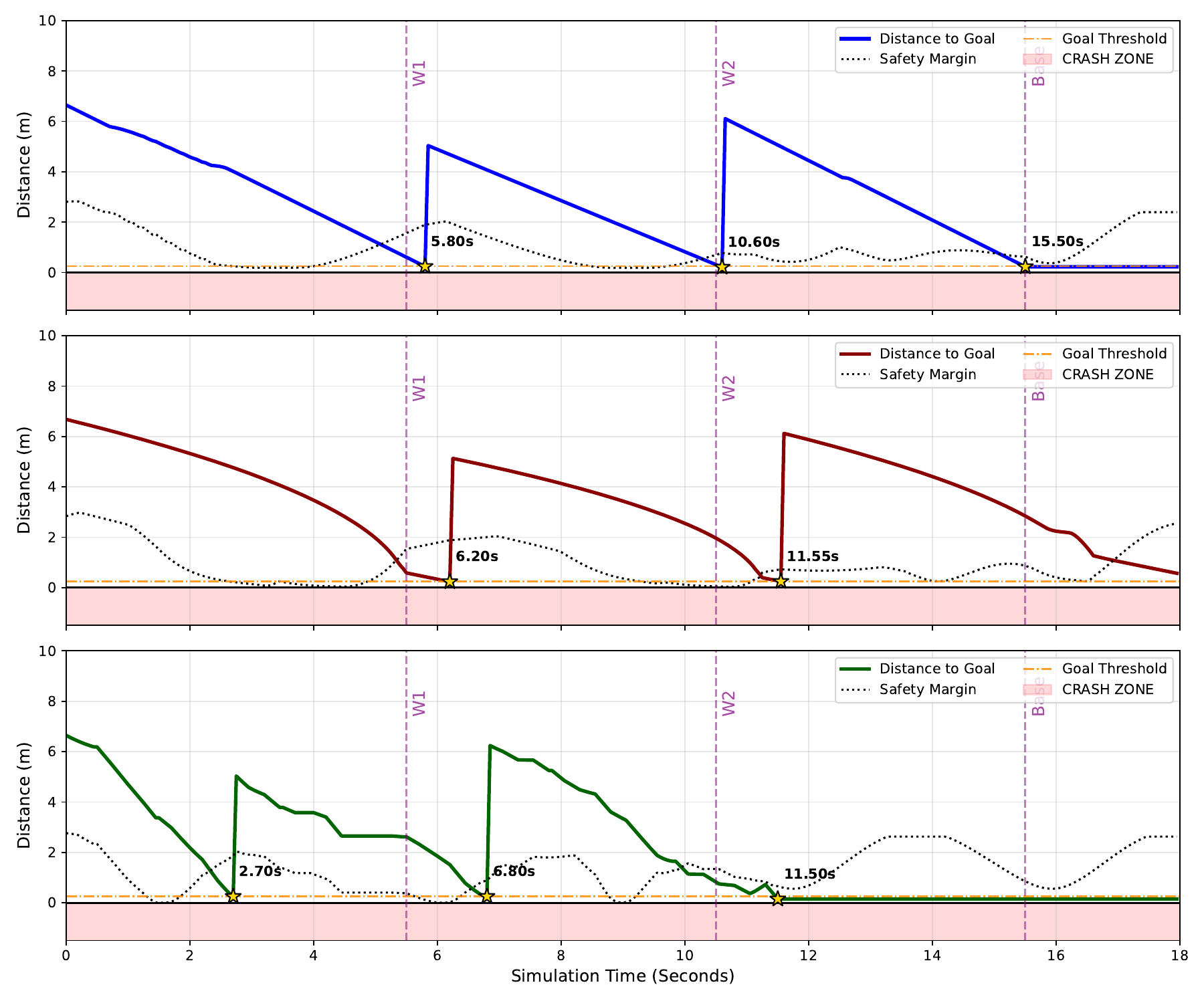}
  \caption{Task progress vs. safety under nominal timing conditions (1
    run). Distance below Goal (yellow dashed) threshold is considered
    meeting the liveness (reachability) goal.}
  \label{fig:res_nominal}
\end{figure}

Under nominal timing conditions, all three architectures successfully
navigate the environment while maintaining a positive safety margin,
keeping the agent out of the crash zone. The proposed geometric
controller, and the TV-CBF baseline both apply soft liveness penality
weights and hence miss the deadline when reaching goals --- as expected
from the weSTL+ semantics (Equation~\eqref{eq:1}). The maximum observed
liveness slack ($\delta_{F}$) of 8.48 is below the maximum theoretical
degradation bound 12.09, thereby validating
Lemma~\ref{thm:inclusive_degradation} for the proposed approach. While
the liveness deadline is violated, both the proposed and TV-CBF maintain
absolute safety. The STL-MPC technique applies maximum actuator velocity
to complete all three reachability goals ahead of time.

\begin{figure}[htbp]
  \centering
  \includegraphics[width=\columnwidth]{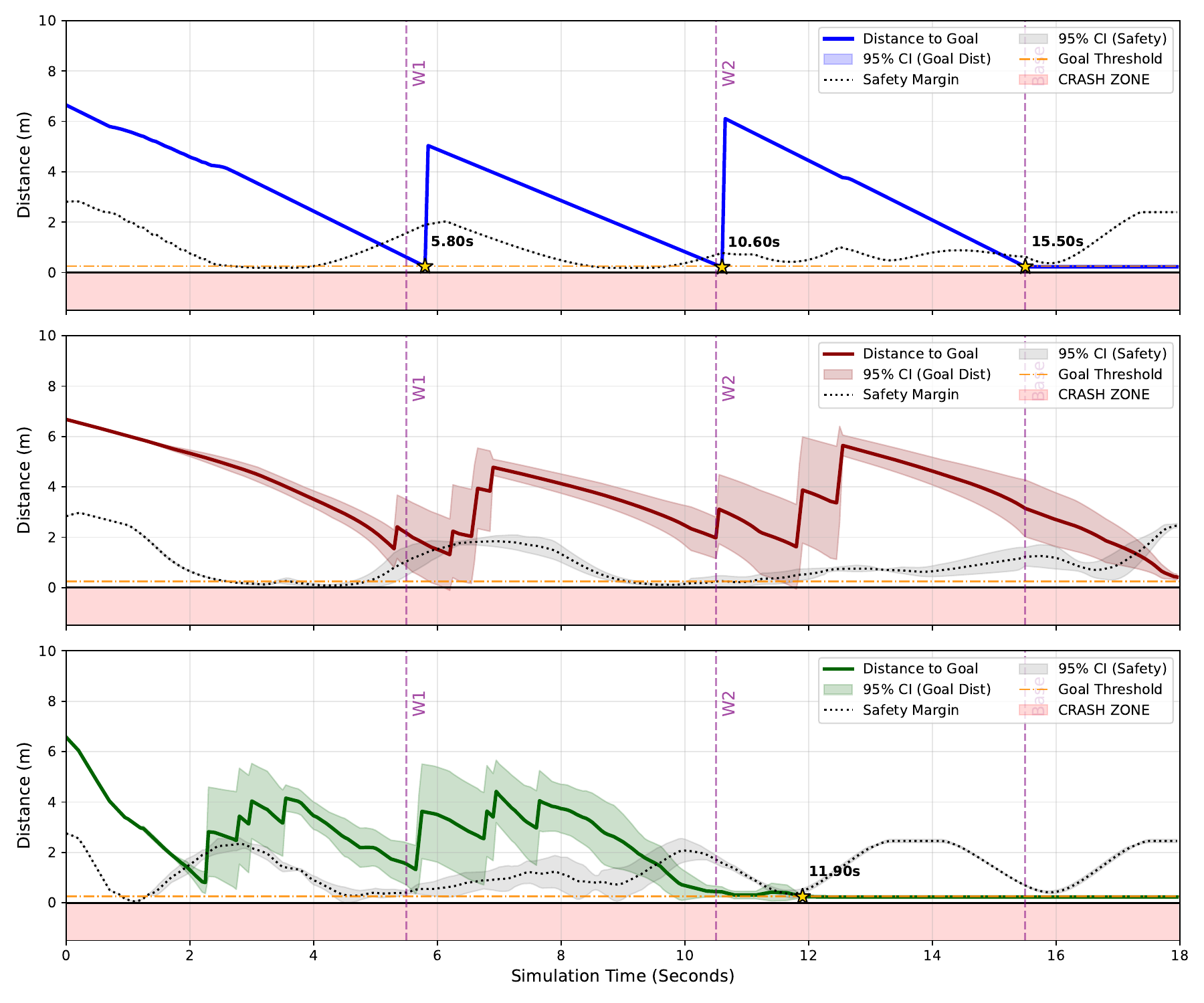}
  \caption{Task progress vs. safety under Snap anomalies (5 runs).
    TV-CBF and STL-MPC are unable to meet goals W1 and W2.}
  \label{fig:res_snap}
\end{figure}

\begin{figure}[htbp]
  \centering
  \includegraphics[width=\columnwidth]{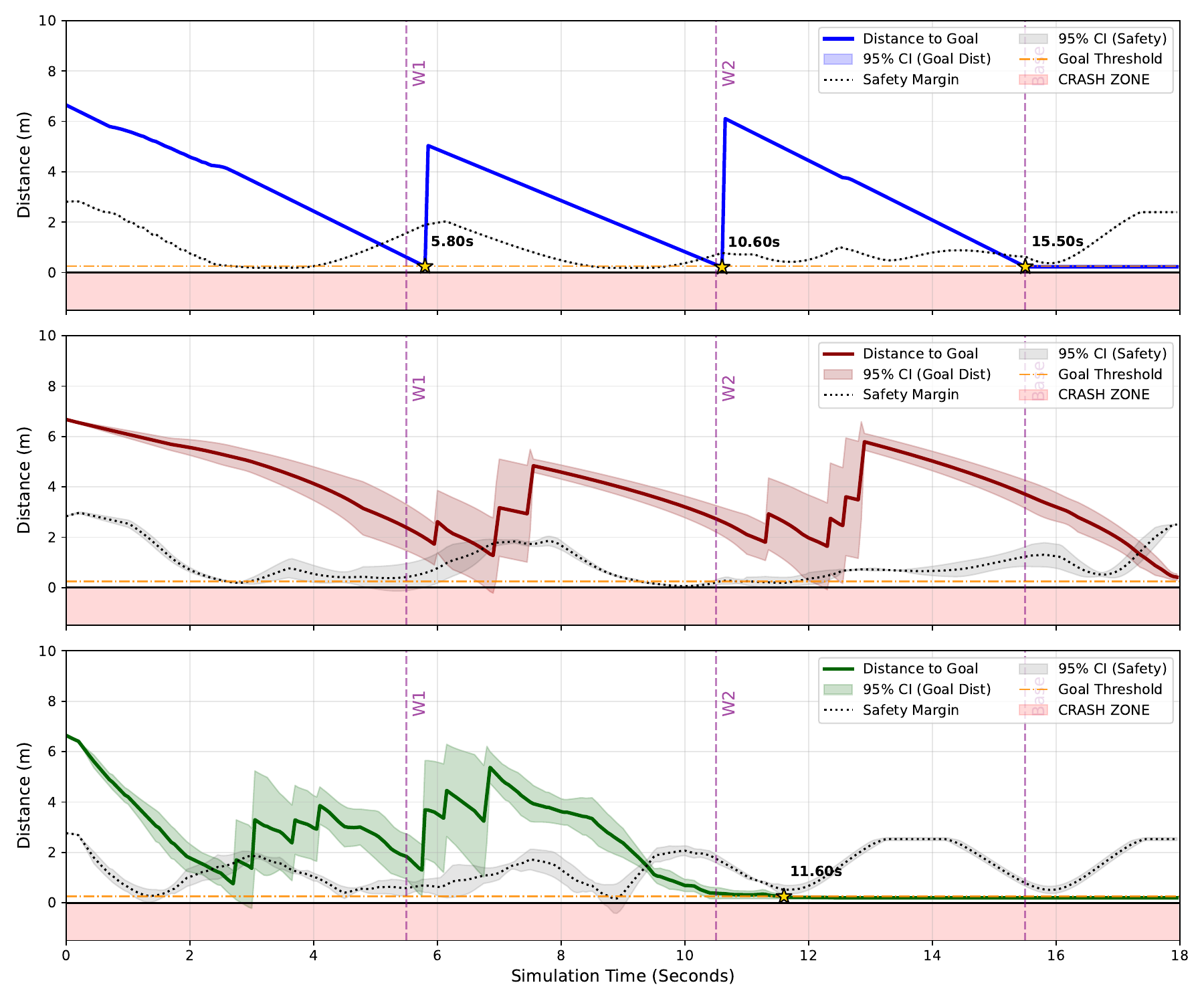}
  \caption{Task progress vs. safety under Offset anomalies (5 runs).
    W1 and W2 reach goals not met by TV-CBF and STL-MPC.}
  \label{fig:res_offset}
\end{figure}

However, during the snap, offset, freeze, and drift experiments, both
the TV-CBF and STL-MPC baselines exhibit severe performance degradation.
Despite the inclusion of the slack variables, the TV-CBF controller's
task progress becomes highly erratic. Across three out of four temporal
anomaly types, the TV-CBF baseline demonstrates significant variance in
its distance to the goal and consistently fails to meet any reachability
goals, though its safety margin generally manages to stay above the
crash zone.

\begin{figure}[htbp]
  \centering
  \includegraphics[width=\columnwidth]{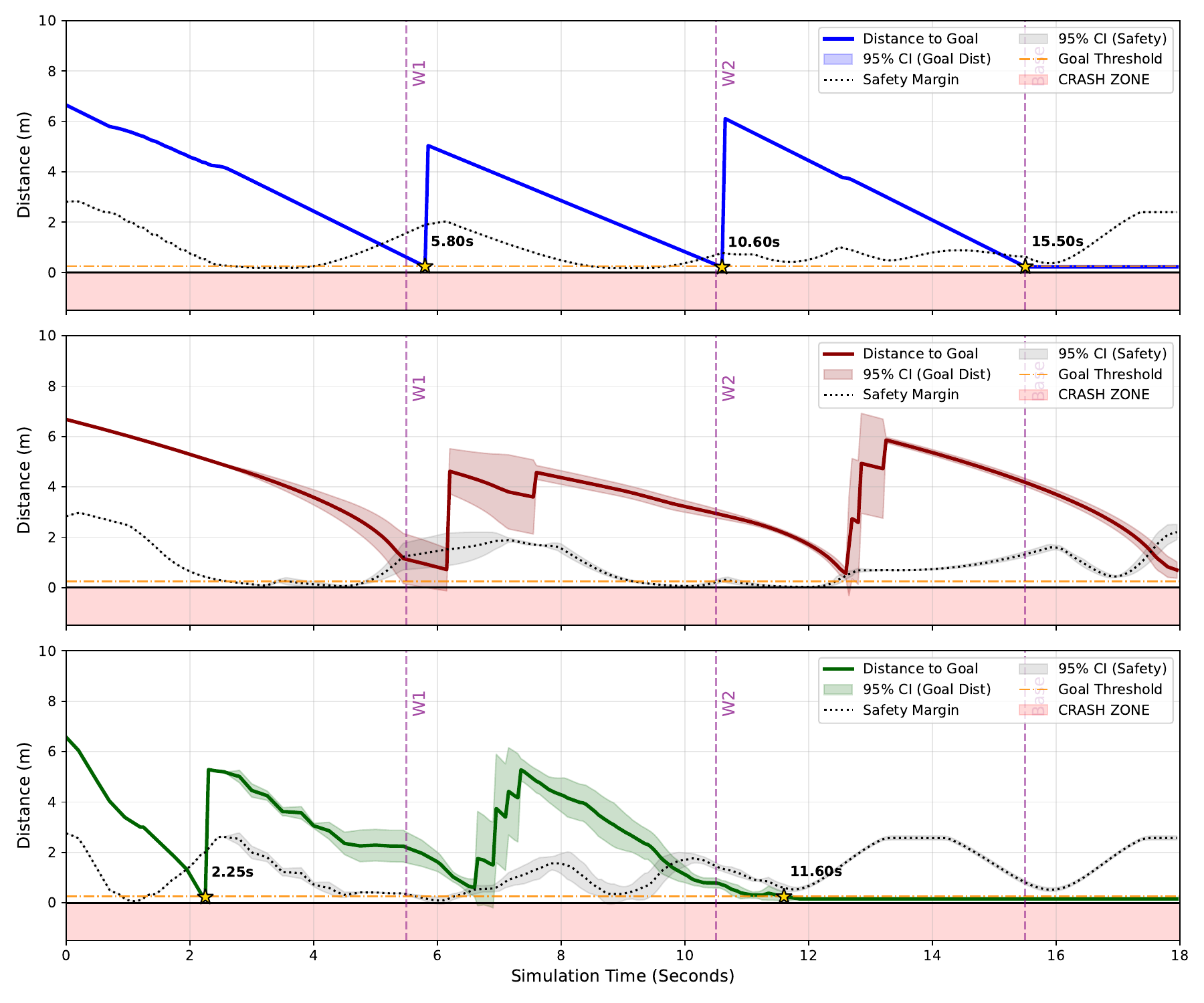}
  \caption{Task progress vs. safety under Freeze anomalies (5 runs).
    Goal not met by TV-CBF. Potential crash in STL-MPC.}
  \label{fig:res_freeze}
\end{figure}

\begin{figure}[htbp]
  \centering
  \includegraphics[width=\columnwidth]{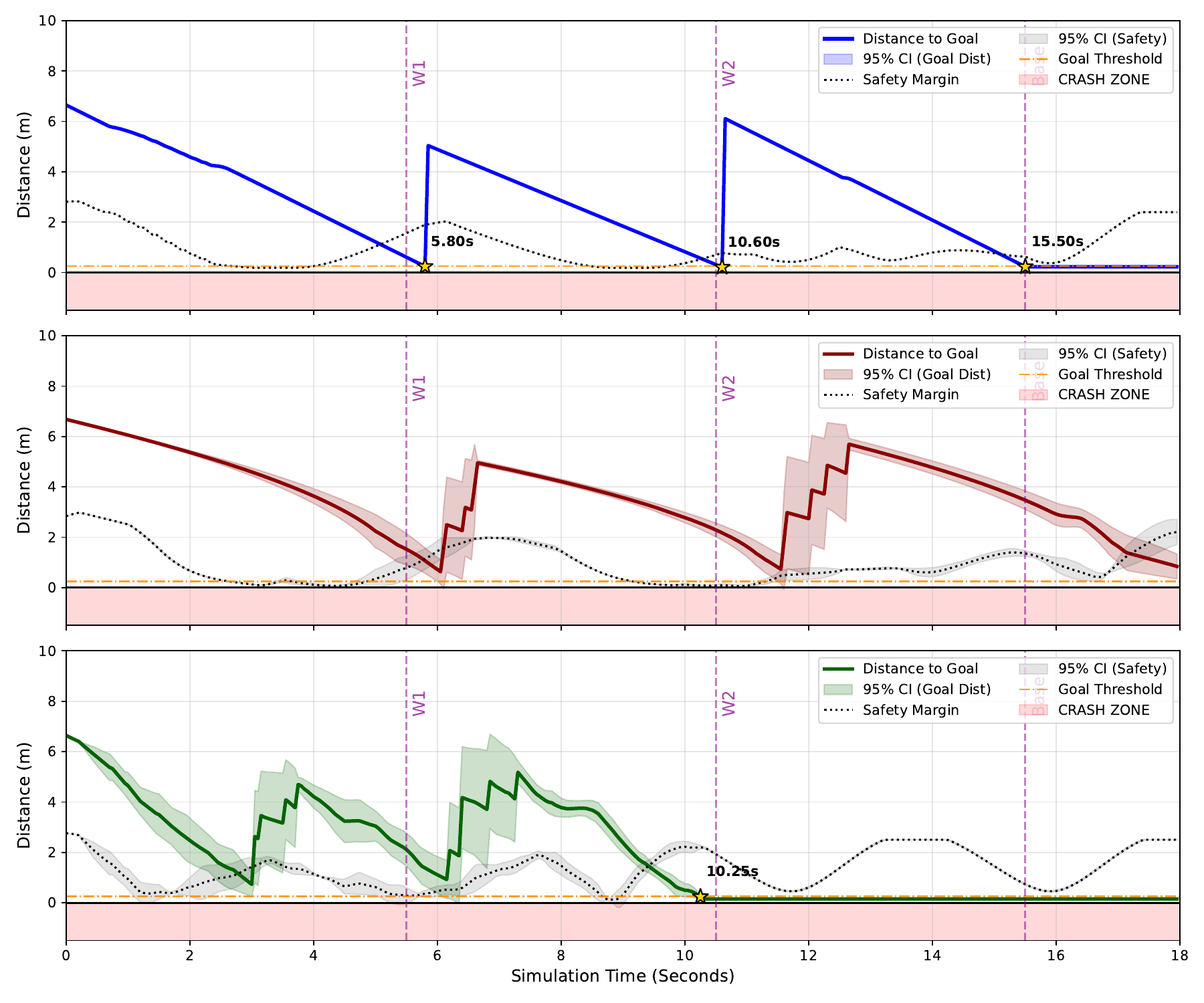}
  \caption{Task progress vs. safety under Drift anomalies (5 runs). Goals
    W1 and W2 not met by STL-MPC with potential for crash.}
  \label{fig:res_drift}
\end{figure}

The STL-MPC baseline also suffers from the same issues as TV-CBF.
Moreover, we also see safety failures. While the reachability slacks
prevent immediate solver crashes, the temporal desynchronization forces
the vehicle into unsafe states. In the snap, offset, freeze, and drift
scenarios, the 95\% confidence interval for the STL-MPC safety margin
drops into the crash zone (safety margin $\leq 0$), indicating recurring
collisions across the runs. Furthermore, its ability to reach targeted
waypoints is also compromised.

In stark contrast, the proposed geometric controller approach succeeds
flawlessly across all anomalous trials. Because the framework converts
global temporal requirements into finite-time level-set inversions
managed entirely through geometric spatial boundaries, its performance
is completely invariant to the simulated clock anomalies. The distance
to the goal and safety margins for remain identical to the nominal
execution, organically adapting to severe clock snaps, offsets, freezes,
and drifts without violating safety bounds or missing deadlines.

\subsection{Execution Time}
\label{sec:execution-time}

To evaluate computational efficiency, we compare the overall median
solve times per control iteration across all evaluated scenarios and
Monte Carlo runs using the Gurobi solver. As summarized in
Table~\ref{tab:execution_time}, the proposed geometric controller
framework achieves an overall median solve time of $3.90\text{ ms}$,
operating well within low-latency bounds for real-time control
execution. In comparison, the TV-CBF baseline requires a median solve
time of $7.39\text{ ms}$. The STL-MPC formulation incurs a significantly
higher computational cost with a median solve time of $492.04\text{ ms}$
per iteration, illustrating the optimization overhead inherent to
mixed-integer encodings over discrete temporal horizons.

\begin{table}[htbp]
  \centering
  \caption{Overall Median Solve Times Across All Scenarios and Runs}
  \label{tab:execution_time}
  \begin{tabular}{|l|c|}
    \hline
    \textbf{Controller / Formulation} & \textbf{Median Solve Time (ms)} \\ \hline
    Proposed Geometric Control  & 3.90 \\ \hline
    TV-CBF  & 7.39 \\ \hline
    STL-MPC  & 492.04 \\ \hline
  \end{tabular}
\end{table}

\section{Related Work}
\label{sec:related_work}

The literature surrounding formal specifications and real-time
controller synthesis for Cyber-Physical Systems (CPS) has evolved to
address various functional gaps, primarily divided into specification
expressiveness and robust control execution.

\subsection{Comparison of Temporal Logics}
\label{sec:comp-temp-logics}

Formalizing complex behaviors in highly dynamic environments
traditionally relies on Signal Temporal Logic
(STL)~\cite{maler2004monitoring}. Standard STL provides strict Boolean
satisfaction and continuous-time robustness metrics but is inherently
bound to rigid global runtime clocks~\cite{maler2004monitoring}. This
structural rigidity renders it insufficient for practical CPS
deployments subject to asynchronous timing anomalies or flexible task
trade-offs.

To address these limitations, recent extensions have branched in two
primary directions. Weighted Signal Temporal Logic (wSTL) and
wSTL+~\cite{mehdipour2021weighted, cardona2023preferences} introduce a
quantitative partial-satisfaction hierarchy, enabling the formulation of
soft specifications scaled by user-defined preference weights. However,
these frameworks still depend heavily on synchronized global clocks.
Conversely, Event-Based Signal Temporal Logic
(eSTL)~\cite{gundana2021event} abandons the rigid global timeline in
favor of asynchronous, event-anchored evaluation, triggering local
property evaluations based on geometric boundaries rather than
time~\cite{gundana2021event}.

The proposed logic in this work, weSTL+, fundamentally bridges these
paradigms. By merging the asynchronous event anchors of eSTL with the
quantitative preference scaling of wSTL+, weSTL+ enables autonomous
agents to user preferred safety and liveness objectives. As demonstrated
in our case study, this allows an autonomous agent to safely react to
unexpected packet delays, clock anomalies, and nested loitering without
losing formal soundness.

\subsection{Comparison of Controller Synthesis}
\label{sec:comp-contr-synth}

While advanced temporal logics effectively specify CPS behaviors, their
downstream control architectures often suffer from systemic fragilities
when deployed in real-world scenarios featuring clock skew, network
delays, clock jitter, and steps/snaps.

Standard time-indexed frameworks, such as Hard
STL-MPC~\cite{raman2014model} and Global Time-Varying Control Barrier
Functions (TV-CBF)~\cite{lindemann2018control}
(Figure~\ref{fig:motivating_example}), rely on explicit local or global
clock variables. As shown in the case study, exposing these controllers
to macroscopic temporal anomalies decouples the digital timeline from
physical state evolution, leading to immediate Quadratic Program (QP)
solver infeasibility or permanent loss of liveness. Relaxed time-indexed
variants, such as Soft Local TV-CBF and
wSTL+-MPC~\cite{cardona2023preferences}, attempt to mitigate these
optimization crashes via soft slack variables (as done in
Section~\ref{sec:experimental-results}). However, under clock anomalies,
these slack buffers are rapidly consumed, causing outright boundary
penetrations into hazards or are unable to meet liveness/reachability
requirements. Dual-layer systems like G\&K-G~\cite{gundana2021event}
gracefully degrade by instantly aborting liveness to maintain safety,
but ultimately fail mission progress.

To circumvent these theoretical and practical vulnerabilities, this
paper introduces a fundamentally timeless geometric controller. By
mapping temporal windows directly into continuous-time geometric
inclusions via finite-time level-set inversion, explicit runtime clock
monitoring is eliminated entirely.

\section{Conclusions}
\label{sec:conclusions}

This paper presented a paradigm shift to geometric control to overcome
the systemic vulnerabilities of traditional, time-indexed Cyber-Physical
Systems (CPS) controllers. By introducing Weighted Event-Based Signal
Temporal Logic (weSTL+), we successfully bridged the gap between
asynchronous event-anchored execution and quantitative preference
scaling. Furthermore, our \textit{sound} two-pass compiler translates
these complex temporal specifications directly into strictly continuous,
$C^1$-differentiable geometric surrogate constraints via finite-time
level-set inversion. As demonstrated in our autonomous robotics case,
and ablation, studies, this approach strictly guarantees the enforcement
of both safety and liveness under severe macroscopic timing anomalies.
Ultimately, shifting from rigid digital timelines to physical spatial
boundaries provides a mathematically sound and highly robust foundation
for resilient real-world CPS deployments.

\bibliographystyle{IEEEtran}
\bibliography{ref.bib}

@inproceedings{donze2010robust,
  title={Robust satisfaction of temporal logic over real-valued signals},
  author={Donz{\'e}, Alexandre and Maler, Oded},
  booktitle={International conference on formal modeling and analysis of timed systems},
  pages={92--106},
  year={2010},
  organization={Springer}
}

@inproceedings{maler2004monitoring,
  title={Monitoring temporal properties of continuous signals},
  author={Maler, Oded and Nickovic, Dejan},
  booktitle={International symposium on formal techniques in real-time and fault-tolerant systems},
  pages={152--166},
  year={2004},
  organization={Springer}
}

@article{gundana2021event,
  title={Event-based signal temporal logic synthesis for single and multi-robot tasks},
  author={Gundana, David and Kress-Gazit, Hadas},
  journal={IEEE Robotics and Automation Letters},
  volume={6},
  number={2},
  pages={3687--3694},
  year={2021},
  publisher={IEEE}
}

@inproceedings{cardona2023preferences,
  title={Preferences on partial satisfaction using weighted signal temporal logic specifications},
  author={Cardona, Gustavo A and Vasile, Cristian-Ioan},
  booktitle={2023 European Control Conference (ECC)},
  pages={1--6},
  year={2023},
  organization={IEEE}
}

@article{mehdipour2021weighted,
  title={Weighted signal temporal logic},
  author={Mehdipour, Noushin and Belta, Calin},
  journal={IEEE Control Systems Letters},
  volume={6},
  pages={1016--1021},
  year={2021},
  publisher={IEEE}
}

@inproceedings{raman2014model,
  title={Model predictive control with signal temporal logic specifications},
  author={Raman, Vasumathi and Donz{\'e}, Alexandre and Maasoumy, Mehdi and Murray, Richard M and Sangiovanni-Vincentelli, Alberto and Seshia, Sanjit A},
  booktitle={53rd IEEE Conference on Decision and Control},
  pages={81--87},
  year={2014},
  organization={IEEE}
}

@article{lindemann2018control,
  title={Control barrier functions for signal temporal logic tasks},
  author={Lindemann, Lars and Dimarogonas, Dimos V},
  journal={IEEE Control Systems Letters},
  volume={3},
  number={1},
  pages={96--101},
  year={2018},
  publisher={IEEE}
}

@inproceedings{lee2008cyber,
  title={Cyber physical systems: Design challenges},
  author={Lee, Edward A},
  booktitle={2008 11th IEEE International Symposium on Object and Component-Oriented Real-Time Distributed Computing (ISORC)},
  pages={363--369},
  year={2008},
  organization={IEEE}
}

@misc{mills2010ntp,
    series = {Request for Comments},
    number = 5905,
    howpublished = {RFC 5905},
    publisher = {RFC Editor},
    doi = {10.17487/RFC5905},
    url = {https://www.rfc-editor.org/info/rfc5905},
    author = {David L. Mills and Jim Martin and Jack Burbank and William Kasch},
    title = {{Network Time Protocol Version 4: Protocol and Algorithms Specification}},
    year = 2010,
    month = jun
}

@article{brunelli2012temperature,
    title = {Temperature compensated time synchronisation in wireless sensor networks},
    author = {Brunelli, D. and Balsamo, D. and Paci, G. and Benini, L.},
    journal = {Electronics Letters},
    volume = {48},
    number = {16},
    year = {2012},
    publisher = {IET},
    doi = {10.1049/el.2012.1246}
}

@inproceedings{cardoso2011network,
    title = {Network latency and packet delay variation in cyber-physical systems},
    author = {Cardoso, J. and Derler, P. and Eidson, J. C. and Lee, E. A.},
    booktitle = {2011 IEEE Network Science Workshop},
    pages = {51--58},
    year = {2011},
    organization = {IEEE}
}

@inproceedings{lisova2017monitoring,
  title={Monitoring of clock synchronization in cyber-physical systems: A sensitivity analysis},
  author={Lisova, Elena and Uhlemann, Elisabeth and {\AA}kerberg, Johan and Bj{\"o}rkman, Mats},
  booktitle={2017 international conference on Internet of Things, embedded systems and communications (IINTEC)},
  pages={134--139},
  year={2017},
  organization={IEEE}
}

@article{ames2016control,
  title={Control barrier function based quadratic programs for safety critical systems},
  author={Ames, Aaron D and Xu, Xiangru and Grizzle, Jessy W and Tabuada, Paulo},
  journal={IEEE transactions on automatic control},
  volume={62},
  number={8},
  pages={3861--3876},
  year={2016},
  publisher={IEEE}
}

@article{bhat2000finite,
  title={Finite-time stability of continuous autonomous systems},
  author={Bhat, Sanjay P and Bernstein, Dennis S},
  journal={SIAM Journal on Control and optimization},
  volume={38},
  number={3},
  pages={751--766},
  year={2000},
  publisher={SIAM}
}

@article{yin2024formal,
  title={Formal synthesis of controllers for safety-critical autonomous systems: Developments and challenges},
  author={Yin, Xiang and Gao, Bingzhao and Yu, Xiao},
  journal={Annual Reviews in Control},
  volume={57},
  pages={100940},
  year={2024},
  publisher={Elsevier}
}

@inproceedings{sadraddini2015robust,
  title={Robust temporal logic model predictive control},
  author={Sadraddini, Sadra and Belta, Calin},
  booktitle={2015 53rd Annual Allerton Conference on Communication, Control, and Computing (Allerton)},
  pages={772--779},
  year={2015},
  organization={IEEE}
}

@inproceedings{guo2012realtime,
  title={Real-time Clock Jump Detection and Repair for Precise Point Positioning},
  author={Guo, Fei and Zhang, Xiaohong},
  booktitle={Proceedings of the 25th International Technical Meeting of the Satellite Division of the Institute of Navigation (ION GNSS 2012)},
  pages={3077--3088},
  year={2012},
  organization={Institute of Navigation}
}

@techreport{ieee1588_2008,
  title={IEEE Standard for a Precision Clock Synchronization Protocol for Networked Measurement and Control Systems (Clause 11)},
  author={{IEEE}},
  institution={IEEE Standards Association},
  number={1588-2008},
  year={2008},
  note={Focus on Clause 11 regarding clock jumps and fault handling}
}

@inproceedings{nist2018improving,
  title={Improving packet synchronization in an NTP server},
  author={Novick, Andrew N and Lombardi, Michael A and Franzen, Kevin and Clark, John},
  booktitle={Proceedings of the 49th Annual Precise Time and Time Interval Systems and Applications Meeting},
  pages={256--260},
  year={2018}
}

@article{liu2020automatic,
  title={Automatic detection of ionospheric scintillation-like GNSS satellite oscillator anomaly using a machine-learning algorithm},
  author={Liu, Yunxiang and Morton, Y. Jade},
  journal={NAVIGATION, Journal of the Institute of Navigation},
  volume={67},
  pages={651--662},
  year={2020},
  publisher={Institute of Navigation},
  doi={10.1002/navi.385}
}

@manual{nxp2025ieee1588,
  title={IEEE 1588 basic overview},
  author={{NXP Semiconductors}},
  year={2025},
  url={https://docs.nxp.com/bundle/AN12149/page/topics/ieee_1588_basic_overview.html}
}

\end{document}